\documentclass[10pt]{article}

\usepackage{graphicx}%
\usepackage{xcolor}
\usepackage{amsmath,amssymb,amsthm,upref,bm}%
\usepackage{labelfig}%

\newtheorem{lemma}{Lemma}%
\newtheorem{corollary}{Corollary}%
\newtheorem{proposition}{Proposition}%
\begin{document}

\baselineskip=4.4mm

\makeatletter

\newcommand{\E}{\mathrm{e}\kern0.2pt} 
\newcommand{\D}{\mathrm{d}\kern0.2pt}
\newcommand{\RR}{\mathbb{R}}
\newcommand{\CC}{\mathbb{C}}%
\newcommand{\ii}{\kern0.05em\mathrm{i}\kern0.05em}

\renewcommand{\Re}{\operatorname{Re}} 
\renewcommand{\Im}{\operatorname{Im}}

\def\bottomfraction{0.9}

\title{\bf Sloshing of a two-layer fluid in a vertical cylinder of constant depth}

\author{Nikolay Kuznetsov and Oleg Motygin}

\date{}

\maketitle

\vspace{-8mm}

\begin{center}
Laboratory for Mathematical Modelling of Wave Phenomena, \\ Institute for Problems
in Mechanical Engineering, Russian Academy of Sciences, \\ V.O., Bol'shoy pr. 61, St
Petersburg 199178, Russian Federation \\ E-mail: nikolay.g.kuznetsov@gmail.com,
mov@ipme.ru
\end{center}

\begin{abstract}
\noindent Sloshing eigenvalues and eigenfunctions are studied for vertical cylinders
of constant, finite depth occupied by a two-layer fluid. Two families of
eigenfrequencies are obtained in the form expressing them explicitly via the
eigenvalues of the Neumann Laplacian in the two-dimensional domain\,---\,cylinder's
cross-section. Eigenfrequencies belonging to one of the families behave similar to
those that describe sloshing in a homogeneous fluid, whereas the other family
includes a large number of sufficiently small frequencies provided the ratio of
densities is close to unity. Various properties of eigenfrequencies are investigated
for cylinders of arbitrary cross-section; they include the dependence on the
interface depth and the ratio of densities, the asymptotics of the eigenvalue
counting function. The behaviour of eigenvalues and the corresponding eigenmodes is
illustrated by numerical examples for circular cylinders without and with a radial
baffle.
\end{abstract}

\setcounter{equation}{0}


\section{Introduction}
\label{sect:1}

Sloshing of various fluids (in particular, consisting of multiple layers of
different density) in containers with rigid walls is a topic of great interest to
engineers, physicists and mathematicians. It has received extensive study; see, for
example, the monographs \cite{FT} and \cite{I} (the second one has more than 100
pages of references). A historical review going back to the 18th century can be
found in \cite{FK}, whereas the comprehensive book \cite{KK} presents an advanced
mathematical approach to the problem based on spectral theory of operators in a
Hilbert space.

During the past decade, sloshing of multiple, immiscible fluids has received a
substantial attention; see, e.g., \cite{M,Mi} and numerous references cited
in the latter article. The importance of these studies for engineering is clearly
outlined in \cite[Sect.~I]{Mi}, where various results obtained so far (mostly
numerical, but also experimental) are described as well. However, there are still
many open questions in this area of research and our aim is to fill in this gap at
least partially.

In the present theoretical study (but with numerous numerical examples), we consider
sloshing of two-layer fluids in containers belonging to a particular class, namely,
cylinders of finite, constant depth that have vertical walls, but arbitrary
cross-section; see Fig.~\ref{fig:1}. Despite this class of containers is rather
specific, nevertheless, it allows us to discover some new effects peculiar to the
process in general. In particular, it is shown that sloshing eigenfrequencies are
distributed more densely in the case of a two-layer fluid than for a homogeneous
one. The reason for this is that there exists two sequences of eigenfrequencies for
such a fluid; those in one of them behave similar to sloshing frequencies in the
case of a homogeneous fluid, whereas the other family includes a large number of
sufficiently small frequencies provided the ratio of densities is close to unity.
The eigenfunctions corresponding to these frequencies have rather specific
behaviour, namely, their velocity potentials have substantial jumps at the
interface, where they also change sign. 

The plan of the paper is as follows. In Subsection \ref{sect:1.1}, the general
sloshing problem is stated for a two-layer fluid (description of geometry and the
governing equations are presented), whereas the corresponding variational principle
is considered in Subsection \ref{sect:1.2}. The case of a vertical cylinder with
horizontal bottom is investigated in detail in Section \ref{sect:2} which consists
of three subsections. Subsection \ref{sect:2.1} deals with reduction based on
separation of variables of the general problem in the case of a vertical cylinder
having constant, finite depth; in particular, a formula for eigenvalues is obtained
along with a description of eigenfunctions. Various properties of eigenvalues are
investigated in Subsection \ref{sect:2.2}; they are illustrated numerically for a
circular cylinder. Properties of sloshing eigenfunctions in a circular cylinder are
considered in Subsection \ref{sect:2.3}. The effect of a radial baffle on sloshing
of a two-layer fluid in a circular cylinder is discussed in Section \ref{sect:3}.
The most important conclusions are presented in Section \ref{sect:4}.

\subsection{General statement of the problem}
\label{sect:1.1}

\begin{figure}[t]
\centering
 \SetLabels
 \L (0.25*0.81) $F$\\
 \L (0.25*0.46) $I$\\
 \L (0.475*0.6) $W_1$\\
 \L (0.475*0.25) $W_2$\\
 \L (0.6*0.81) $x_1$\\
 \L (0.34*0.78) $x_2$\\
 \L (0.515*0.98) $y$\\
 \L (0.97*0.62) $h$\\
 \L (-0.008*0.44) $d$\\
 \endSetLabels
 \leavevmode\AffixLabels{\includegraphics[width=55mm]{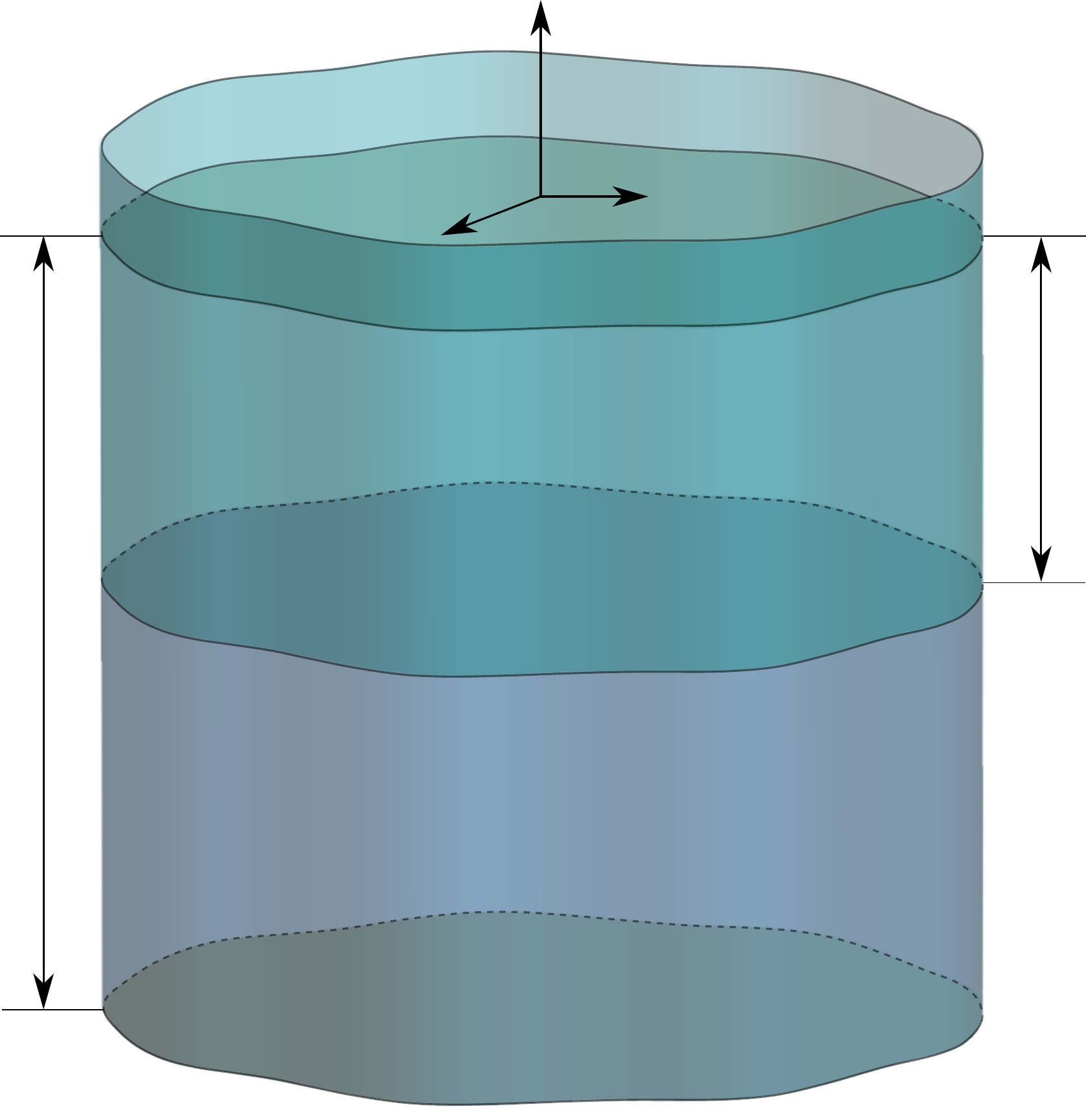}}
 \caption{A sketch of cylinder's geometry and notation; in particular, $F$ and $I$
 denote the free surface of the upper (lighter) fluid and the interface, respectively;
 the latter separates immiscible fluids of different densities.}
 \label{fig:1}
\end{figure}

Let two immiscible, inviscid, incompressible, heavy fluids occupy an open container,
whose walls and bottom are rigid surfaces. Choosing rectangular Cartesian
coordinates $(x_1, x_2, y)$ so that their origin belongs to the mean free surface of
the upper fluid and the $y$-axis is directed upwards, the whole fluid domain $W$
lies in the lower half-space
\[ \{ -\infty < x_1, x_2 < +\infty, \, y < 0 \} .
\]
The boundary $\partial W$ is assumed to be piece-wise smooth and such that every two
adjacent smooth pieces of $\partial W$ are not tangent along their common edge. We
also suppose that each horizontal cross-section of $W$ is a bounded two-dimensional
domain; that is, a connected, open set in the corresponding plane. (The latter
assumption is made for the sake of simplicity because it excludes the possibility of
two or more interfaces between fluids at different levels.) The free surface $F$
bounding above the upper fluid of density $\rho_1 > 0$ is the non-empty interior of
$\partial W \cap \{ y=0 \}$. Let $h > 0$ be less than $\max \{ |y|: \, (x_1, x_2, y)
\in \partial W \}$, then the interface $I = W \cap \{y = -h\}$ separates the upper
fluid from the lower one of density $\rho_2 > \rho_1$. We denote by $W_1$ and $W_2$
the domains $W \cap \{y > -h\}$ and $W \cap \{y < -h\}$ occupied by the upper and
lower fluids, respectively. The surface tension is neglected and we suppose the
fluid motion to be irrotational and of small amplitude. This allows us to linearize
the boundary condition on $F$ and the coupling condition on $I$. With a
time-harmonic factor, say $\cos \omega t$, removed, the velocity potentials $u^{(1)}
(x_1, x_2, y)$ and $u^{(2)} (x_1, x_2, y)$ (they may be taken to be real functions)
describing the flow in $W_1$ and $W_2$, respectively, must satisfy the following
coupled boundary value problem:
\begin{align}
&  u^{(j)}_{x_1 x_1} + u^{(j)}_{x_2 x_2} + u^{(j)}_{yy} = 0 \quad {\rm in} \ W_j,
\quad j=1,2, \label{lap} \\ 
& u^{(1)}_y = \nu u^{(1)} \quad {\rm on} \ F,
\label{nuf} \\ 
& \rho \bigl( u^{(2)}_y - \nu u^{(2)} \bigr) = u^{(1)}_y - \nu
u^{(1)} \quad {\rm on} \ I, \label{nui} \\
& u^{(2)}_y = u^{(1)}_y \quad {\rm on} \
I, \label{yi} \\ 
& \partial u^{(j)}/\partial n = 0 \quad {\rm on} \ B_j\quad j=1,2.
\label{nc}
\end{align}
Here $\rho=\rho_2/\rho_1 >1$ is the non-dimensional parameter characterizing the
stratification, whereas the spectral parameter $\nu$ is equal to $\omega^2/g$, where
$\omega$ is the radian frequency of the water oscillations and $g$ is the
acceleration due to gravity; $B_j=\partial W_j\setminus (\bar F \cup \bar I)$ is the
rigid boundary of $W_j$. Combining \eqref{nui} and \eqref{yi}, one obtains another
form of the spectral coupling condition \eqref{nui}:
\begin{equation}
(\rho -1) u^{(2)}_y = \nu \bigl( \rho u^{(2)} - u^{(1)} \bigr)
\quad {\rm on}\ I. \label{nui2}
\end{equation}
It is convenient to complement relations \eqref{lap}--\eqref{nc} by the orthogonality
conditions
\begin{equation}
\int_F u^{(1)} \,\D x = 0 \quad {\rm and} \quad \int_I \bigl( \rho u^{(2)} - u^{(1)}
\bigr)\,\D x = 0, \quad \D x = \D x_1 \D x_2 , \label{ort}
\end{equation}
thus excluding the zero eigenvalue of \eqref{lap}--\eqref{nc}.

If $\rho =1$, then conditions \eqref{nui} and \eqref{yi} mean that $u^{(1)}$ and
$u^{(2)}$ are harmonic continuations of each other across the interface $I$. In this
case, problem \eqref{lap}--\eqref{nc} complemented by the first orthogonality
condition \eqref{ort} (the second one is trivial) moves to the problem that
describes sloshing of a homogeneous fluid in $W$. It is well-known since the 1950s
that the latter problem has a discrete spectrum; that is, there exists a sequence
$\{ \nu_n^W \}_1^\infty$ of positive eigenvalues of finite multiplicity (the
superscript $W$ is used here and below to distinguish the sloshing eigenvalues when
a homogeneous fluid occupies the whole domain $W$, from those corresponding to a
two-layer fluid which will be denoted simply by $\nu_n$). In this sequence, the
eigenvalues are written in increasing order and repeated according to their
multiplicity; moreover, $\nu_n^W \to \infty$ as $n \to \infty$. The sequence of the
corresponding eigenfunctions $\{ u_n \}_1^\infty$ belongs to the Sobolev space $H^1
(W)$ and forms a complete system in an appropriate Hilbert space. These results can
be found in many sources, for example, in the book \cite{KK}. Moreover, the problem
describing sloshing of a multi-layer fluid is also considered in \cite[\S~3.6.4]{KK}
as an application of the general approach developed in this book. It is established
that its spectrum is discrete; in particular, the spectrum of
\eqref{lap}--\eqref{nc} complemented by \eqref{ort} is discrete.

\subsection{Variational principle}
\label{sect:1.2}

It is well known that the problem describing sloshing of a homogeneous fluid in a
bounded domain $W$ can be presented as a variational problem and the corresponding
Rayleigh quotient is as follows:
\begin{equation}
R_W (u) = \frac{\int_W |\nabla u|^2 \,\D x \D y}{\int_F u^2 \,\D x}. \label{Rayhom}
\end{equation}
For obtaining the fundamental eigenvalue $\nu_1^W\!$, one has to minimize $R_W (u)$
over the subspace of the Sobolev space $H^1 (W)$ consisting of functions that
satisfy the first orthogonality condition \eqref{ort}. A standard procedure yields
the whole sequence of eigenvalues.

In the case of a two-layer fluid, the Rayleigh quotient for the sloshing problem has
the following form:
\begin{equation}
R (u^{(1)},u^{(2)}) = \frac{\int_{W_1}\! \left| \nabla u^{(1)} \right|^2 \, \D x \D y
+ \rho \int_{W_2}\! \left| \nabla u^{(2)} \right|^2 \, \D x \D y }{\int_F \left[
u^{(1)} \right]^2 \, \D x + (\rho - 1)^{-1} \int_I \left[ \rho u^{(2)} - u^{(1)}
\right]^2 \, \D x} . \label{Raytwo}
\end{equation}
To determine the fundamental sloshing eigenvalue $\nu_1$, one has to minimize $R
(u^{(1)},u^{(2)})$ over the subspace of $H^1 (W_1) \oplus H^1 (W_2)$, where both
orthogonality conditions \eqref{ort} hold. Again, the usual procedure allows us to
find $\nu_n$ for all $n > 1$.

Now we are in a position to prove the following assertion in which the assumption
that~$W$ is bounded is essential.

\begin{proposition}
\label{prop1}
Let\/ $\nu_1^W$ and\/ $\nu_1$ be the fundamental eigenvalues of the sloshing problem in
the bounded domain\/ $W$ for homogeneous and two-layer fluids respectively. Then the
inequality\/ $\nu_1 < \nu_1^W$ holds.
\end{proposition}

\begin{proof}
If $u_1$ is an eigenfunction corresponding to $\nu_1^W$, then \eqref{Rayhom} takes
the form:
\[ \nu_1^W = \frac{\int_W |\nabla u_1|^2\,\D x \D y}{\int_F u_1^2\,\D x} \, . \]
Let $u^{(1)}$ and $u^{(2)}$ be the restrictions of $\rho u_1$ and $u_1$ to $W_1$ and
$W_2$, respectively. Then the pair $\left( u^{(1)},\,u^{(2)} \right)$ is an
admissible element for the Rayleigh quotient \eqref{Raytwo}. Substituting it into
\eqref{Raytwo}, we obtain that
\[ R (\rho u_1, u_1) = \frac{\int_{W_1} \left| \nabla u_1 \right|^2\,\D x \D y + 
\rho^{-1} \int_{W_2} \left| \nabla u_1 \right|^2\,\D x \D y }{\int_F u_1^2\,\D x}.
\]
Comparing this equality with the previous one and taking into account that $\rho
>1$, one finds that $R (\rho u_1, u_1) < \nu_1^W$. Since $\nu_1$ is the minimum of
\eqref{Raytwo}, we conclude that $\nu_1 < \nu_1^W$.
\end{proof}

\section{Vertical cylinder with horizontal bottom}
\label{sect:2}

\subsection{Reduction of problem (\ref{lap})--(\ref{nc})}
\label{sect:2.1}

\begin{samepage}
Let $W = \{ x = (x_1,x_2) \in D, \, y \in (-d,0) \}$, where $D$ is a bounded
two-dimensional domain (the horizontal cross-section of the vertical cylinder $W$)
with piecewise smooth boundary, and the positive $d < \infty$ is the constant depth
of $W$. Thus, the cylinder's bottom $\{ x \in D, \, y = -d \}$ is horizontal,
whereas the free surface and the interface are 
\[ F = \{ x \in D, \, y=0 \} \ \ \mbox{and} \ \ I =
\{ x \in D, \, y=-h \} ,
\]
respectively, where $0<h<d$.

\end{samepage}

For a homogeneous fluid occupying such a container, the sloshing problem is
equivalent to the free membrane problem. Indeed, putting
\[ u (x,y) = v (x) \cosh k (y+d) , 
\]
one reduces problem \eqref{lap}--\eqref{nc} with $\rho = 1$, complemented by the
first orthogonality condition \eqref{ort} to the following spectral problem:
\begin{equation}
\nabla_x^2 v + k^2 v = 0 \ \ \mbox{in} \ \ D , \quad \partial v / \partial n_x = 0 \
\ \mbox{on} \ \ \partial D , \ \ \int_D v \, \D x = 0. \label{fm}
\end{equation}
Here $\nabla_x = (\partial/\partial x_1 , \partial/\partial x_2)$ and $n_x$ is a
unit normal to $\partial D$ in ${\RR}^2$. It is clear that $\nu^W$ is an eigenvalue
of the sloshing problem if and only if $k^2$ is an eigenvalue of \eqref{fm} and
\begin{equation}
\nu^W = k \tanh kd . \label{nuW}
\end{equation} 
It is well-known that problem \eqref{fm} has a sequence of positive eigenvalues $\{
k_n^2 \}_1^\infty$ written in increasing order and repeated according to their
finite multiplicity, and such that $k_n^2\to \infty$ as $n\to \infty$. The
corresponding eigenfunctions form a complete system in $H^1 (D)$.

Letting $d \to \infty$ in formula \eqref{nuW}, one obtains $\nu^W = k$ for the
infinitely deep cylinder $W = D \times (-\infty, 0)$. Indeed,
\begin{equation}
u (x,y) = v (x) \, \E^{k y} \label{nuW'}
\end{equation} 
is the sloshing eigenfunction in this cylinder if and only if $k^2$ is an eigenvalue
of problem~\eqref{fm} and $v$ is the corresponding eigenfunction. Furthermore, the
same is true when this cylinder is occupied by a two-layer fluid; this immediately
follows by verifying conditions \eqref{lap}--\eqref{nc} and \eqref{ort} for the
function defined by formula \eqref{nuW'}.

Let us turn to considering a reduction procedure similar to the described above in
the case when $W$ has a finite depth and is occupied by a two-layer fluid. We put
\begin{align}
& u^{(1)} (x,y) = v (x) \, A\, \bigl[ k\cosh (k y) + \nu \sinh (k y) \bigr] ,
\label{u1} \\ & u^{(2)} (x,y) = v (x) \, B \cosh k(y+d) , \label{u2}
\end{align} 
where $A$ and $B$ are constants that do not vanish simultaneously, and so these
functions satisfy conditions \eqref{nuf} and \eqref{nc}, respectively. In this way,
we reduce problem \eqref{lap}--\eqref{nc} and \eqref{ort}, to \eqref{fm}, but
instead of the explicit formula \eqref{nuW} it is complemented by the quadratic
equation
\begin{equation}
Q (\nu, k) = 0 \, , \label{qe}
\end{equation} 
where
\begin{multline}
Q (\nu, k) = \nu^2 \bigl[\sinh(kh)\sinh(k(d - h))+\rho \cosh(k h)\cosh(k(d -
h))\bigr] - \nu \rho k \sinh(k d) \\ {} + k^2(\rho-1) \sinh(kh)\sinh(k(d - h)) \, ,
\quad k > 0 ,
 \label{Qdef}
\end{multline}
Indeed, the roots of \eqref{qe} define sloshing eigenvalues provided $k^2$ is an
eigenvalue of \eqref{fm}, because $Q (\nu, k)$ is proportional to the determinant
of the following linear algebraic system for $A$ and~$B$:
\begin{gather}
A \bigl(k^2 - \nu^2\bigr) \sinh(kh) + B \rho \bigl[k \sinh(k(d - h))- \nu \cosh(k(d
- h))\bigr] =0, \notag\\
A \bigl[k \sinh(kh) - \nu \cosh(kh)\bigr] + B \sinh(k(d - h)) = 0 . \notag
\end{gather}
This homogeneous system arises when one substitutes expressions \eqref{u1} and
\eqref{u2} into the coupling conditions \eqref{nui} and \eqref{yi}, and so it
defines eigensolutions of the sloshing problem if there exists its non-trivial
solution. Therefore, the determinant must vanish, thus implying that $\nu$ is a
solution of equation \eqref{qe}.

{\it Hence\/ $\nu$ is an eigenvalue of the two-layer sloshing problem in the cylinder\/
$W$ if and only if\/ $\nu$ satisfies\/ \eqref{qe}, where\/ $k^2$ is an eigenvalue of\/
\eqref{fm} in the cylinder's cross-section\/ $D$.}

\def\Th{t_1}
\def\Tdmh{t_2}

Furthermore, equation \eqref{qe} is equivalent to the following one
\begin{equation}
\nu^2 (\rho + \Th\Tdmh) -
 \nu \rho k (\Th+\Tdmh) +
 k^2(\rho-1) \Th\Tdmh = 0 \, , \label{qe2}
\end{equation}
whose coefficients are symmetric functions of
\[ \Th = \tanh (k h) , \quad \Tdmh = \tanh(k(d - h)) \, .
\]
Indeed, the left-hand side of \eqref{qe2} is equal to $Q (\nu, k)/\bigl[ \cosh (k
h)\cosh (k(d - h))\bigr]$, because
\[ \sinh (k d) = \cosh (k h) \sinh (k(d - h)) + \cosh (k(d - h)) \sinh (k h) . \]
It is clear that \eqref{qe2} is a significant simplification of \eqref{qe}.

Let us show that both roots of \eqref{qe2}, say
\begin{equation}
\nu^{(\pm)} = \frac{k \bigl(b \pm \sqrt{\mathcal{D}}\bigr)}{2 (\rho + \Th \Tdmh)}
\label{nupm}
\end{equation}
are positive (for the sake of brevity, their dependence on $k$, $d$, $h$ and $\rho$
is omitted). Here $b = \rho (\Th + \Tdmh) > 0$, and since $\rho > 1$,
\begin{equation}
\mathcal{D} = b^2 - 4 (\rho - 1) \Th \Tdmh (\rho + \Th \Tdmh ) < b^2 ,
\label{Ddef2}
\end{equation}
On the other hand, we have
\[ \mathcal{D} = \rho^2 (\Th - \Tdmh)^2  + 4 [ \rho \Th \Tdmh (1 - \Th \Tdmh) +
\Th^2 \Tdmh^2 ] > 0 .
\]
These inequalities prove that \eqref{nupm} is positive. Now we are in a position to
formulate the following.

\begin{proposition}\label{prop2} 
Let\/ $W$ be a vertical cylinder of the constant depth\/ $d$ and the uniform
cross-section\/ $D$. If\/ $W$ is occupied by a two-layer fluid with the density ratio\/
$\rho > 1$ and the interface at the depth\/ $h$, then formula\/ \eqref{nupm} defines
two sequences of sloshing eigenvalues\/ $\bigl\{ \nu_n^{(+)} \bigr\}_1^\infty$ and\/
$\bigl\{ \nu_n^{(-)} \bigr\}_1^\infty$, where\/ $k = k_n > 0$ and $k_n^2$ is an
eigenvalue of problem\/ \eqref{fm}. Also, the inequality\/ $\nu_n^{(-)} < \nu_n^{(+)}$
holds for every\/ $n = 1,2,\dots$.

Moreover, $u^{(1)}_n$ and\/ $u^{(2)}_n$---\,the sloshing modes in\/ $W_1$ and\/ $W_2$,
respectively, that correspond to\/ $\nu_n^{(\pm)}$---\,are defined by formulae\/
\eqref{u1} and\/ \eqref{u2}, respectively, in which\/ $A_n$ is an arbitrary non-zero
real constant and
\[ B_n = A_n \frac{\nu_n^{(\pm)} \cosh(k_n h) - k_n \sinh(k_n h)}{\sinh(k_n (d - h))}
\, ,
\]
whereas\/ $v = v_n$ is any eigenfunction of problem\/ \eqref{fm} corresponding
to\/ $k^2_n$.
\end{proposition}

\subsection{Properties of eigenvalues}
\label{sect:2.2}

The fact that the coefficients of equation \eqref{qe2} are symmetric functions of
$\Th$ and $\Tdmh$ has the following important consequence.

\begin{proposition}\label{prop3} 
Let a vertical cylinder\/ $W$ of the constant depth\/ $d$ and the uniform cross-section\/
$D$ be occupied by a two-layer fluid with the density ratio\/ $\rho > 1$, then the
sequences\/ $\bigl\{ \nu_n^{(+)} \bigr\}_1^\infty$ and\/ $\bigl\{ \nu_n^{(-)}
\bigr\}_1^\infty$ serve also as the set of sloshing eigenvalues in the case when the
depths of the subdomains\/ $W_1$ and\/ $W_2$, respectively, are equal to\/ $d-h$ and\/
$h$, respectively.
\end{proposition}

In other words, if $\rho > 1$ and $d > 0$ are fixed, whereas $h \in (0, d)$, then
the graph of $\nu_n^{(\pm)} (h)$ is symmetric about the line $h = d/2$. Also,
formula \eqref{nupm} implies that
\[ \lim_{h \to +0} \nu^{(-)}_n (h) =  0 \, , \quad \lim_{h \to +0} \nu^{(+)}_n (h) =
k_n \tanh (k_n d)
\]
and
\[ \nu^{(\pm)}_n (d / 2) = k_n \frac{\tanh (k_n d / 2)}{\rho + \tanh^2 (k_n d / 2)}
\left[ \rho \pm \sqrt{\rho \bigl[ 1 - \tanh^2 (k_n d / 2) \bigr] + \tanh^2 (k_n d /
2)} \right] ,
\]
where $k_n^2$ is the $n$th eigenvalue of problem \eqref{fm}.

Let us analyse the behaviour of $\nu^{(-)}$ and $\nu^{(+)}$, which are clearly
continuous functions of~$\rho$. A straightforward calculation yields that
\[ \frac{\partial \nu^{(-)}}{\partial \rho} = 
\frac{k \, \Th \Tdmh \left[ \sqrt{\mathcal{D}} (\Th+\Tdmh) + 2 \, \Th\Tdmh
(1+\Th\Tdmh) + \rho \left( 2 - \Th^2 - \Tdmh^2 \right) \right]}{2 \sqrt{\mathcal{D}}
(\rho + \Th\Tdmh)^2} > 0 \, ,
\]
because $\Th, \Tdmh < 1$. Hence, $\nu^{(-)}$ monotonically increases as $\rho$ goes
from unity to infinity. Since formula \eqref{Ddef2} implies that
\begin{equation}
\sqrt \mathcal{D} \sim b \ \ \mbox{as} \ \rho \to 1 + 0 \ \ \mbox{and} \ \ \sqrt
\mathcal{D} \sim \rho \, |\Th - \Tdmh| \ \ \mbox{as} \ \rho \to \infty ,
\label{sim}
\end{equation}
\eqref{nupm} with the lower sign yields that
\[ \lim_{\rho \to 1 + 0} \nu^{(-)} = 0 \quad \mbox{and} \quad \lim_{\rho \to \infty}
\nu^{(-)} = \frac{k}{2} (\Th+\Tdmh - |\Th-\Tdmh|) \, ,
\]
respectively. It is clear that $\Th + \Tdmh - |\Th - \Tdmh| = 2 \min \{\Th,
\Tdmh\}$, and so
\[ \lim_{\rho \to \infty} \nu^{(-)} = k \, \min \{\Th, \Tdmh\} < k \, \tanh k d =
\nu^W \, .
\]


In the same way, we obtain
\[ \frac{\partial \nu^{(+)}}{\partial \rho} = 
\frac{k \, \Th \Tdmh \left[ \sqrt{\mathcal{D}} (\Th+\Tdmh) - 2 \, \Th\Tdmh
(1+\Th\Tdmh) - \rho \left( 2 - \Th^2 - \Tdmh^2 \right) \right]}{2 \sqrt{\mathcal{D}}
(\rho + \Th\Tdmh)^2} \, .
\]
Let us show that the expression in the square brackets is negative. Indeed, we have
\[ \mathcal{D} (\Th+\Tdmh)^2 - \left[ 2 \, \Th\Tdmh (1+\Th\Tdmh) +
\rho \left( 2 - \Th^2 - \Tdmh^2 \right) \right]^2 = - 4 \bigl(1 - \Th^2\bigr)\bigl(1
- \Tdmh^2\bigr) (\rho + \Th\Tdmh )^2 < 0 \, ,
\]
which implies that
\[ \frac{\partial \nu^{(+)}}{\partial \rho} < 0 \, .
\]
Hence, $\nu^{(+)}$ monotonically decreases as $\rho$ goes from unity to infinity.
Using \eqref{sim} in \eqref{nupm} with the upper sign, we see that
\[ \lim_{\rho \to 1 + 0} \nu^{(+)} = k \frac{\Th + \Tdmh}{1 + \Th \Tdmh} = k \tanh(kd)
\ \ \mbox{and} \ \ \lim_{\rho \to \infty} \nu^{(+)} = \frac{k}{2} (\Th + \Tdmh +
|\Th - \Tdmh|) \, .
\]
Since $\Th + \Tdmh + |\Th - \Tdmh| = 2 \max \{\Th, \Tdmh\}$, we have that
\[ \lim_{\rho \to \infty} \nu^{(+)} = k \, \max \{\Th, \Tdmh\} < k \, \tanh k d =
\nu^W \, . 
\]

Let us summarize the obtained above as follows.

\begin{proposition}\label{prop4}
For every\/ $n = 1,2,\dots$, the sloshing eigenvalue\/ $\nu_n^{(+)}$ $\bigl[ \nu_n^{(-)}
\bigr]$ regarded as a function of\/ $\rho \in (1, \infty)$ monotonically decreases
{\rm[}increases, respectively\/{\rm]}. Its range is the nonempty interval
\begin{align*}
& \bigl( k_n \max \{\tanh (k_n h) , \, \tanh (k_n (d - h))\}, \ k_n \tanh (k_n d)
\bigr) \\ & \bigl[ \bigl( 0, \ k_n \min \{\tanh (k_n h) , \, \tanh (k_n (d - h))\}
\bigr) , \ \textit{respectively} \bigr] .
\end{align*}
Both intervals shrink away as\/ $h$ tends to zero.
\end{proposition}

Combining this proposition and formula \eqref{nuW}, we arrive at the following.

\begin{corollary}\label{corol1} 
For every\/ $n = 1,2,\dots$, the inequalities\/ \mbox{$\nu_n^{(-)} < \nu_n^{(+)} <
\nu_n^W$} are valid for sloshing eigenvalues in a two-layer fluid.
\end{corollary}

Letting $n \to \infty$ in formula \eqref{nupm}, it is straightforward to obtain the
following.

\begin{lemma}\label{lemma1}
For every $\rho > 1$ the sequences of sloshing eigenvalues have the following
asymptotic behaviour:
\[ \nu_n^{(+)} \sim k_n , \quad \nu_n^{(-)} \sim \frac{\rho - 1}{\rho+1} \, k_n 
\quad as \ n \to \infty .
\]
Here\/ $k^2_n$ is an eigenvalue of problem\/ \eqref{fm}, whereas the remainder terms are
exponentially small.
\end{lemma}

Notice that the asymptotic behaviour of $\nu_n^{(+)}$ as $n \to \infty$ coincides
with that of $\nu_n^W$; this immediately follows from formula \eqref{nuW}. Since
$\nu_n^{(-)}$ has a different behaviour as $n \to \infty$, the distribution function
$\mathcal{N} (\nu)$ (it is equal to the number of sloshing eigenvalues of both kinds
that are less than $\nu > 0$) better characterizes the asymptotics of the whole
spectrum in the case when a two-layer fluid with the density ratio $\rho > 1$
occupies the cylinder $W$. The so-called Weyl's law, describing the distribution of
eigenvalues of problem~\eqref{fm}, is used for evaluating $\mathcal{N} (\nu)$.

Let $\mathcal{N}_D (\lambda)$ denote the distribution function for problem
\eqref{fm}; that is, the number of eigenvalues of this problem that are less than
$\lambda$. Then Weyl's law (see \cite[p.~442]{CH}) says that
\[ \mathcal{N}_D (\lambda) \sim \lambda |D| / (4 \pi) \ \ \mbox{as} \ \lambda \to 
\infty \, ,
\]
where $|D|$ stands for the area of the domain $D$. Combining Weyl's law and
Lemma~\ref{lemma1}, we are in a position to obtain $\mathcal{N} (\nu)$ as $\nu \to
\infty$, but prior to that it is worth mentioning that {\it when a homogeneous fluid
occupies\/ $W$, the behaviour of the distribution function for\/ $\nu_n^W$ is as
follows}:
\[ \mathcal{N}_W (\nu) \sim \nu^2 |D| / (4 \pi) \ \ as \ \nu \to \infty .
\]
This is a consequence of Weyl's law and formula \eqref{nuW}, which expresses
sloshing eigenvalues via those of problem \eqref{fm}.

By Lemma~\ref{lemma1}, the number of eigenvalues $\nu_n^{(+)}$ that are less than $\nu$ is
asymptotically equal~to
\[ \nu^2 |D| / (4 \pi) \ \ \mbox{as} \ \nu \to \infty \, ,
\]
which coincides with the behaviour of $\mathcal{N}_W (\nu)$. On the other hand,
Lemma 1 yields that the analogous number for $\nu_n^{(-)}$ is asymptotically equal
to
\[ \frac{\nu^2 |D|}{4 \pi} \left[ \frac{\rho + 1}{\rho - 1} \right]^2 \ \ \mbox{as}
\ \nu \to \infty \, .
\]
Summing up these asymptotics, we arrive at the following.

\begin{proposition}\label{prop5}
Let\/ $W$ be a vertical cylinder of the constant depth having the uniform cross-section\/
$D$. If\/ $W$ is occupied by a two-layer fluid with the density ratio\/ $\rho > 1$,
then the distribution function of sloshing eigenvalues has the following behaviour\/
\begin{equation}
\mathcal{N} (\nu) \sim \frac{\nu^2 |D| (\rho^2 + 1)}{2 \pi (\rho - 1)^2} \ \
\mbox{as} \ \nu \to \infty \, ,
\label{eq:Nasympt}
\end{equation}
independent of the values\/ $d$ and\/ $h$.
\end{proposition}

It is clear that $\mathcal{N} (\nu) \to \nu^2 |D| / (2 \pi)$ as $\rho \to \infty$.
Hence, $\mathcal{N} (\nu)$ is, roughly speaking, two times greater than
$\mathcal{N}_W (\nu)$ when $\nu$ and $\rho$ are sufficiently large; that is,
sloshing eigenvalues are distributed more densely for a two-layer fluid with large
$\rho$ than for a homogeneous one. Since $\mathcal{N} (\nu) \to \infty$ as $\rho \to
1 + 0$, it is reasonable to expect that sloshing eigenvalues are distributed even
more densely for a two-layer fluid when $\rho - 1$ is close to zero. It is worth
mentioning that \eqref{eq:Nasympt} is a particular case of Theorem 2.4.1, proven in
\cite{Putin2022} for container geometries that are not cylindrical.

\begin{figure}[t]
\centering
 \SetLabels
 \L (0.956*-0.02) $\rho$\\
 \L (0.035*0.706) $\nu^{(-)}_{143,...,12}$\\
 \L (-0.018*0.53) {\footnotesize144}\\
 \L (0.41*0.055) {\footnotesize1}\\
 \L (0.39*0.155) {\footnotesize2}\\
 \L (0.375*0.24) {\footnotesize3}\\
 \L (0.365*0.285) {\footnotesize4}\\
 \L (0.34*0.43) {\footnotesize5}\\
 \L (0.333*0.48) {\footnotesize6}\\
 \L (0.315*0.59) {\footnotesize7}\\
 \L (0.305*0.635) {\footnotesize8}\\
 \L (0.3*0.686) {\footnotesize9}\\
 \L (0.275*0.75) {\footnotesize10}\\
 \L (0.25*0.81) {\footnotesize11}\\
 \L (0.115*0.365) {\footnotesize\bfseries1}\\
 \L (0.1*0.65) {\footnotesize\bfseries2}\\
 \L (0.095*0.82) {\footnotesize\bfseries3}\\
 \L (0.09*0.9) {\footnotesize\bfseries4}\\
 \endSetLabels
 \leavevmode\AffixLabels{\includegraphics[width=125mm]{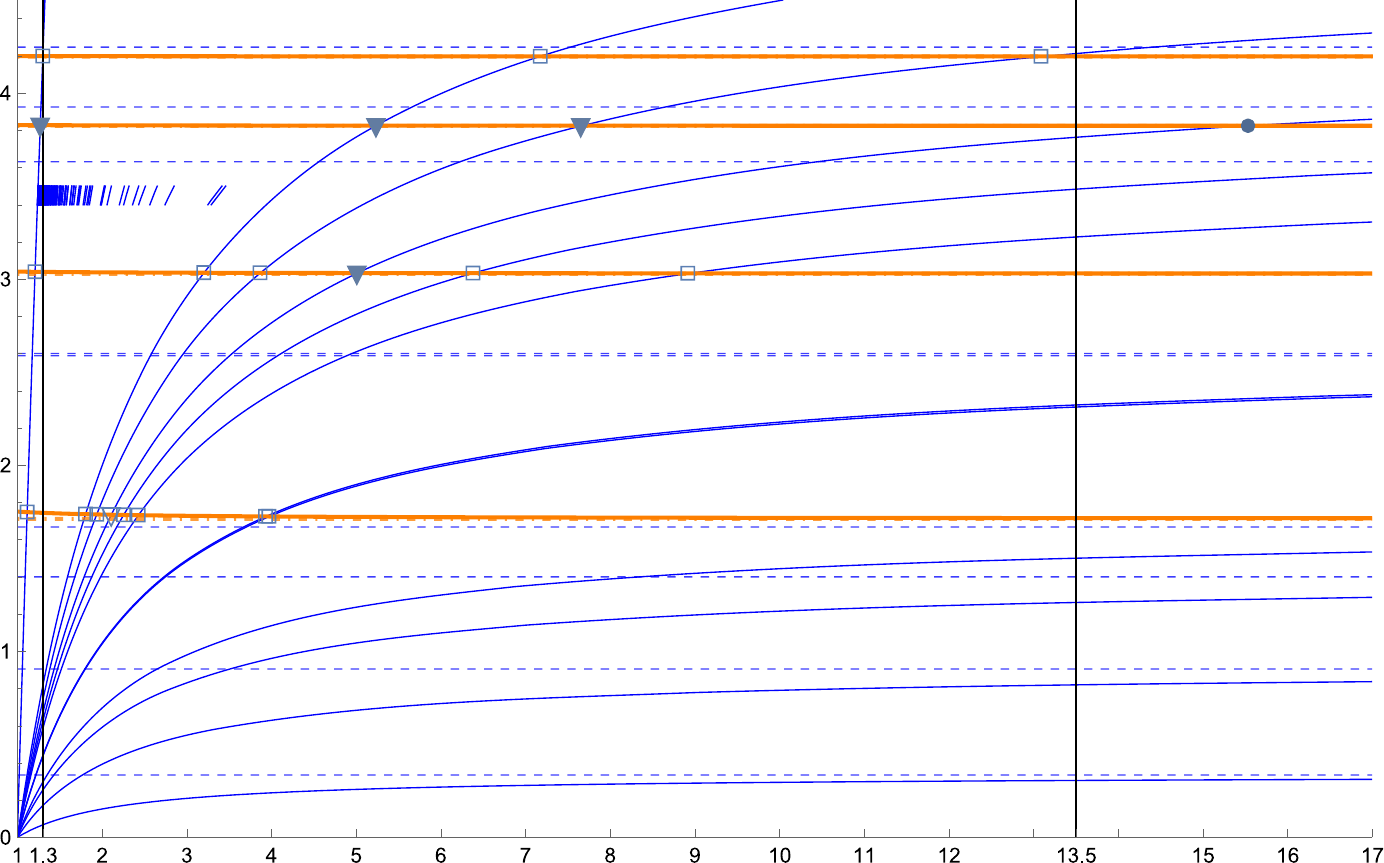}}
\vspace{2mm}
 \caption{For the circular cylinder with $D = \{ |x| < 1 \}$, $d = 1$ and $h =
1/10$, the graphs $\nu^{(+)}_{1,...,4} (\rho)$ (numbered in bold) and
$\nu^{(-)}_{1,...,11,144} (\rho)$ are plotted for $\rho \in [0, 17]$, whereas the
positions of $\nu^{-}_{12,...,143} (\rho)$ are just located between $\nu^{(+)}_2$
and $\nu^{{+}}_3$. The corresponding asymptotes are shown by dashed lines. Points of
intersection of graphs are marked according to the multiplicity of the corresponding
eigenvalue which is two (bullet), three (triangle) or four (square).}
 \label{fig:2}
\end{figure}

For the circular cylinder, whose cross-section is the unit disk, $d = 1$ and $h =
1/10$ (or $h = 9/10$, what is the same according to Proposition~\ref{prop3}), this
is illustrated in Fig.~\ref{fig:2}. In this figure, the subscripts of the sequence
$\{ \bar{k}_m \}_1^\infty$, that include all positive zeros of $J_\mu'$ (here
$J_\mu$ is a Bessel function whose order $\mu$ is any nonnegative integer) arranged
in ascending order, are used for numbering. Thus, the first eleven elements of this
sequence are as follows:
\begin{align}
& \bar{k}_1 = j'_{1,1} \approx 1.841183... , &\!\!\!& \bar{k}_2 = j'_{2,1} \approx
3.054236... , &\!\!\!& \bar{k}_3 = j'_{0,1} \approx 3.831705... ,\nonumber \\
&\bar{k}_4 = j'_{3,1} \approx 4.201188... , &\!\!\!& \bar{k}_5 = j'_{4,1} \approx
5.317553... , &\!\!\!& \bar{k}_6 = j'_{1,2} \approx 5.331442... , \nonumber \\
&\bar{k}_7 = j'_{5,1} \approx 6.415616... , &\!\!\!& \bar{k}_8 = j'_{2,2} \approx
6.706133..., &\!\!\!& \bar{k}_9 = j'_{0,2} \approx 7.015586... , \nonumber \\
&\bar{k}_{10} = j'_{6,1} \approx 7.501266... , &\!\!\!& \bar{k}_{11} = j'_{3,2}
\approx 8.015236... , &\!\!\!& \label{zeroes}
\end{align}
whereas $\bar{k}_{143} = j'_{0,10} \approx 32.189679... $, $\bar{k}_{144} =
j'_{14,5} \approx 32.236969...$ Here $j'_{\mu,s}$ stands for the $s$th positive
zero of $J'_\mu$ (for $\mu = 0$ this numbering differs from that used in \cite{AS},
where $j'_{0,s}$ is the $s$th nonnegative zero). Notice that $\{ \bar{k}_m
\}_1^\infty$, unlike $\{ k_n \}_1^\infty$, does not take into account multiplicity
of eigenvalues of problem \eqref{fm}, and this distinguishes these sequences.


In Fig.~\ref{fig:2}, the graphs $\nu_m^{(+)} (\rho)$ and $\nu_m^{(-)} (\rho)$ (see formula
\eqref{nupm} with $k = \bar{k}_m$) are plotted for $m = 1,2,3,4$ and $m =
1,2,\dots,11,144$, respectively; the vertical lines at $\rho = 1.3$ and $\rho =
13.5$ correspond to the gasoline/water and water/mercury superpositions,
respectively.

We see that each $\nu_m^{(+)} (\rho)$ rapidly asymptotes to the value $\bar{k}_m
\tanh (\frac{9}{10} \bar{k}_m)$ as $\rho \to \infty$, and the first ten values
$\nu_m^{(-)}$ are interlaced with the first four $\nu_m^{(+)}$ for water superposed
over mercury. On the other hand, the number of values $\nu_m^{(-)}$ interlacing with
the first four $\nu_m^{(+)}$ is $144$ for gasoline/water and the sequence of the
interlacing values is as follows:
\[ \{ \nu^{(-)}_{1}, \dots, \nu^{(-)}_{31}, \nu^{(+)}_{1}, \nu^{(-)}_{32}, \dots,
\nu^{(-)}_{78}, \nu^{(+)}_{2}, \nu^{(-)}_{79}, \dots, \nu^{(-)}_{119}, \nu^{(+)}_{3},
\nu^{(-)}_{120}, \dots, \nu^{(-)}_{144}, \nu^{(+)}_{4}, \nu^{(-)}_{145}, \dots \} .
\]
Here
\begin{gather*}
\nu^{(-)}_{1} \approx 0.067182...,\quad \nu^{(-)}_{31} \approx 1.718432...,\quad
\nu^{(+)}_{1} \approx 1.744132...,\quad \nu^{(-)}_{32} \approx 1.786995...,\\
\nu^{(-)}_{78} \approx 3.002318...,\quad \nu^{(+)}_{2} \approx 3.039090...,\quad
\nu^{(-)}_{79} \approx 3.069011...,\\ \nu^{(-)}_{119} \approx 3.798636...,\quad
\nu^{(+)}_{3} \approx 3.827601...,\quad \nu^{(-)}_{120} \approx 3.829050...,\\
\nu^{(-)}_{144} \approx 4.197291...,\quad \nu^{(+)}_{4} \approx 4.199020...,\quad
\nu^{(-)}_{145} \approx 4.207027...
\end{gather*}

\begin{figure}[t]
\centering
 \SetLabels
 \L (0.956*-0.02) $\rho$\\
 \L (0.033*0.715) $\nu^{(-)}_{142,...,8}$\\
 \L (-0.018*0.53) {\footnotesize143}\\
 \L (0.52*0.205) {\footnotesize1}\\
 \L (0.54*0.465) {\footnotesize2}\\
 \L (0.57*0.64) {\footnotesize3}\\
 \L (0.6*0.72) {\footnotesize4}\\
 \L (0.66*0.96) {\footnotesize5}\\
 \L (0.64*0.995) {\footnotesize6}\\
 \L (0.29*0.96) {\footnotesize7}\\
 \L (0.125*0.355) {\footnotesize\bfseries1}\\
 \L (0.09*0.655) {\footnotesize\bfseries2}\\
 \L (0.085*0.832) {\footnotesize\bfseries3}\\
 \L (0.08*0.91) {\footnotesize\bfseries4}\\
 \endSetLabels
 \leavevmode\AffixLabels{\includegraphics[width=125mm]{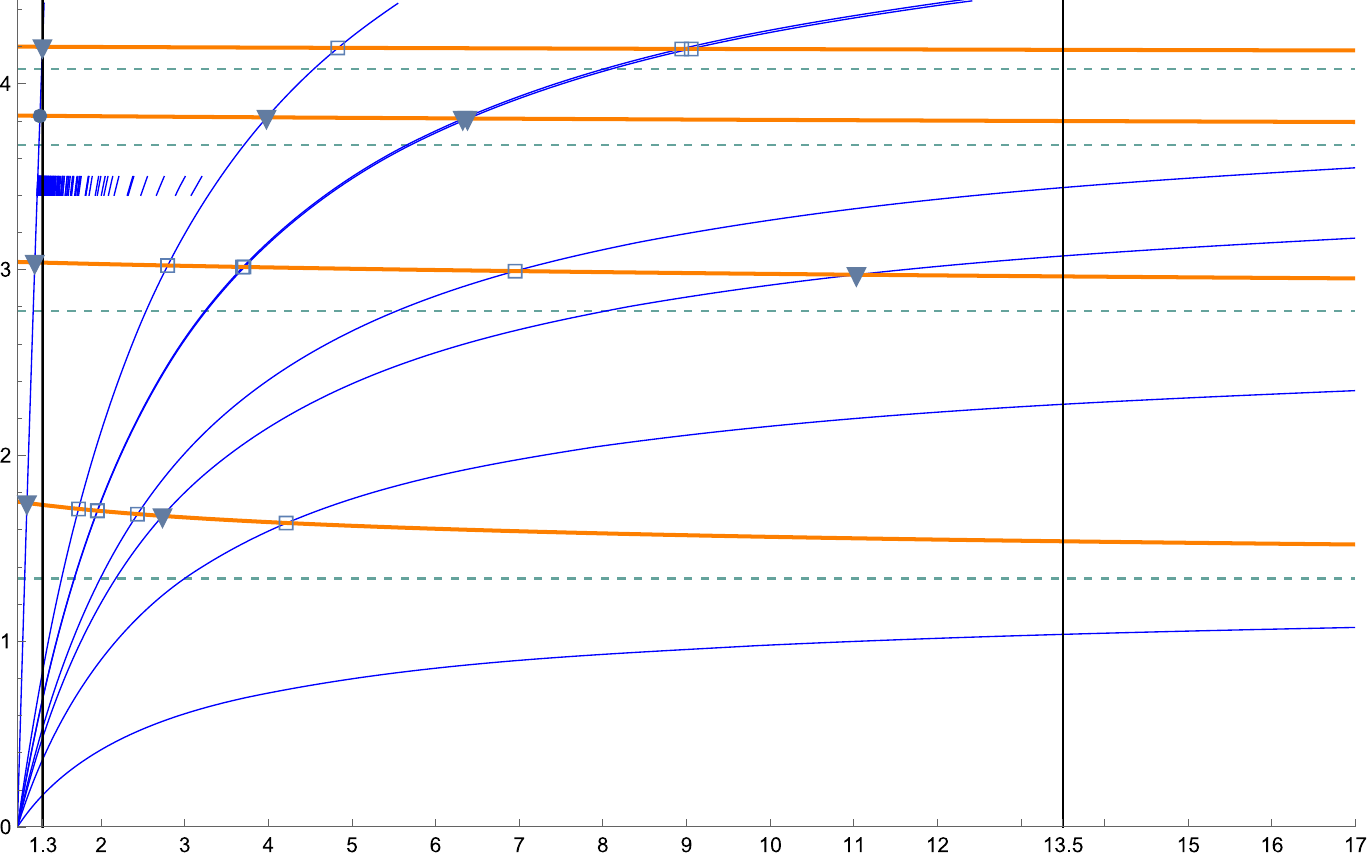}}
 \vspace{2mm}
 \caption{For the circular cylinder with $D = \{ |x| < 1 \}$, $d = 1$ and $h =
1/2$, the graphs $\nu^{(+)}_{1,...,4} (\rho)$ (numbered in bold) and
$\nu^{(-)}_{1,...,7,143} (\rho)$ are plotted for $\rho \in [0, 17]$, whereas the
positions of $\nu^{(-)}_{8,...,142} (\rho)$ are just located between $\nu^{(+)}_2$
and $\nu^{(+)}_3$. The corresponding asymptotes are shown by dashed lines. Points of
intersection of graphs are marked according to the multiplicity of the corresponding
eigenvalue which is two (bullet), three (triangle) or four (square).}
 \label{fig:3}
\end{figure}

Furthermore, when the interface is located at $h = 1/2$, computations show that the
number of values $\nu_m^{(-)}$ that are interlaced with the first four $\nu_m^{(+)}$
is slightly less for both marked values of $\rho$; see Fig.~\ref{fig:3}. In
particular, only four $\nu_m^{(-)}$ are interlaced with the first four
$\nu_m^{(+)}$ for $\rho = 13.5$. On the other hand, the number of values
$\nu_m^{(-)}$ interlacing with the first four $\nu_m^{(+)}$ is $143$ for
gasoline/water and the sequence of the interlacing values is as follows:
\[ \{ \nu^{(-)}_{1}, \dots, \nu^{(-)}_{27}, \nu^{(+)}_{1}, \nu^{(-)}_{28}, \dots,
\nu^{(-)}_{78}, \nu^{(+)}_{2}, \nu^{(-)}_{79}, \dots, \nu^{(-)}_{119}, \nu^{(+)}_{3},
\nu^{(-)}_{120}, \dots, \nu^{(-)}_{143}, \nu^{(+)}_{4}, \nu^{(-)}_{144}, \dots \} .
\]
It almost coincides with that given above. Here
\begin{gather*}
\nu^{(-)}_{1} \approx 0.169347...,\quad
\nu^{(-)}_{27} \approx 1.717867...,\quad
\nu^{(+)}_{1} \approx 1.733024...,\quad
\nu^{(-)}_{28} \approx 1.737866...,\\
\nu^{(-)}_{78} \approx 3.034963...,\quad
\nu^{(+)}_{2} \approx 3.037086...,\quad
\nu^{(-)}_{79} \approx 3.099223...,\\
\nu^{(-)}_{119} \approx 3.811116...,\quad
\nu^{(+)}_{3} \approx 3.827088...,\quad
\nu^{(-)}_{120} \approx 3.841065...,\\
\nu^{(-)}_{143} \approx 4.198653...,\quad
\nu^{(+)}_{4} \approx 4.198760...,\quad
\nu^{(-)}_{144} \approx 4.204822... 
\end{gather*}

An interesting feature of the case $h = 1/2$ is that both $\nu_m^{(+)} (\rho)$ and
$\nu_m^{(-)} (\rho)$ asymptote to the same value $\bar{k}_m \tanh (\frac{1}{2}
\bar{k}_m)$ from above and from below, respectively, as $\rho \to \infty$, but they
approach to this value rather slowly which distinguish this case from that, where $h
= 1/10$.

Let us turn to another characterisation of eigenvalues, namely, their multiplicity.
Recall that if a vertical cylinder of finite depth is occupied by a homogeneous
fluid, then the multiplicity of every sloshing eigenvalue coincides with that of the
corresponding eigenvalue of problem \eqref{fm}.  In the case of a circular cylinder,
every eigenvalue is either simple or has multiplicity two.

When a two-layer fluid is sloshing in a circular cylinder, all eigenvalues are also
either simple or have multiplicity two on every curve belonging either to $\bigl\{
\nu^{(+)}_m (\rho) \bigr\}_1^\infty$ or to $\bigl\{ \nu^{(-)}_m (\rho)
\bigr\}_1^\infty$, except for points of intersection of these curves; see
Figs.~\ref{fig:2} and \ref{fig:3}, where these points are marked according to the
multiplicity of the corresponding eigenvalue which is two (bullet), three (triangle)
or four (square). Most of these points have multiplicity four, whereas multiplicity
two is rather rare. There are no points of intersection on $\nu^{(-)}_1 (\rho),
\dots, \nu^{(-)}_4 (\rho)$ for all $\rho > 0$ when $h = 1/10$, but only $\nu^{(-)}_1
(\rho)$ has this property when $h = 1/2$.

Finally, an important property of the two-layer sloshing concerns the lowest
eigenvalue; namely, it is much smaller than that in the case of homogeneous fluid.
Indeed, in the example presented in Fig.~\ref{fig:2}, the lowest eigenvalue is
$\nu_{1} = \nu^{(-)}_{1} \approx 0.067182...$ for $\rho = 1.3$, which is less than
$4 \%$ of the value $1.750797...$\,---\,the lowest eigenvalue when a homogeneous
fluid is sloshing in the same cylinder. Moreover, the lowest eigenvalue is less than
$\bar{k}_1 \tanh (\bar{k}_1 / 10) \approx 0.335216...$ for all $\rho > 0$, and
$\nu_{1} (\rho) = \nu^{(-)}_{1} (\rho)$ approaches this value from below as $\rho
\to \infty$; see the bottom of Fig.~\ref{fig:2}. In the example presented in
Fig.~\ref{fig:3}, the lowest eigenvalue is $\nu_{1} = \nu^{(-)}_{1} \approx
0.169347...$ for $\rho = 1.3$; this is about $10 \%$ of the lowest eigenvalue in the
case of a homogeneous fluid. Again, the lowest eigenvalue is less than $\bar{k}_1
\tanh (\bar{k}_1 / 2) \approx 1.337025...$ for all $\rho > 0$; see the bottom of
Fig.~\ref{fig:3}.

\subsection{Properties of sloshing eigenfunctions in a circular cylinder}
\label{sect:2.3}

\begin{figure}[p]
\centering
 \SetLabels
 \L (0.92*-0.02) $y$\\
 \L (0.91*0.93) {\small $U^{(1)}_{11,-}$}\\
 \L (0.57*0.27) {\small $U^{(2)}_{11,-}$}\\
 \L (0.46*0.65) {\small $U^{(2)}_{1,+}$}\\
  \L (0.91*0.78) {\small $U^{(1)}_{1,+}$}\\
 \endSetLabels
 \leavevmode\AffixLabels{\includegraphics[width=80mm]{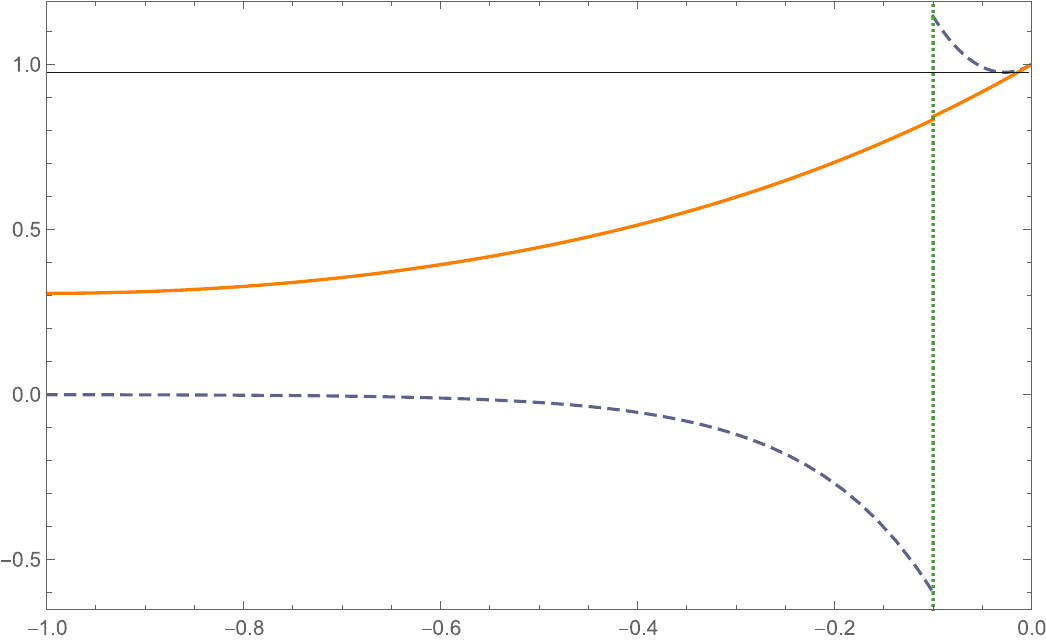}}
 \caption{If $h = 1/10$ in the cylinder with $D = \{ |x| < 1 \}$ and $d = 1$, then for
 $\rho \approx 1.804998...$ we have $\nu_1^{(+)} (\rho) = \nu_{11}^{(-)} (\rho) \approx
 1.736930...$ For the corresponding eigenfunctions \eqref{eigen}, the amplitude factors
 $U^{(1,2)}_{1,+} (y)$ (solid lines) and $U^{(1,2)}_{11,-} (y)$ (dashed lines) are
 plotted.}
\label{fig:4}
%
\vspace{2.4mm}
%
\centering
 \SetLabels
 \L (0.92*-0.02) $y$\\
 \L (0.895*0.935) {\small $U^{(1)}_{8,-}$}\\
 \L (0.47*0.13) {\small $U^{(2)}_{8,-}$}\\
 \L (0.55*0.47) {\small $U^{(2)}_{2,+}$}\\
  \L (0.91*0.77) {\small $U^{(1)}_{2,+}$}\\
 \endSetLabels
 \leavevmode\AffixLabels{\includegraphics[width=80mm]{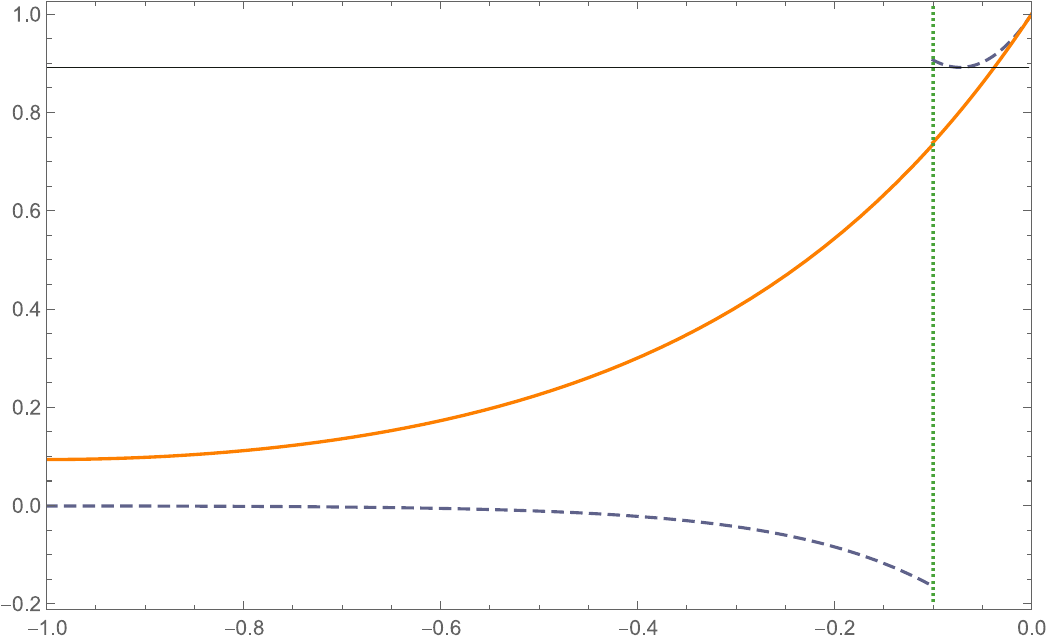}}
 \caption{If $h = 1/10$ in the cylinder with $D = \{ |x| < 1 \}$ and $d = 1$, then for
 $\rho \approx 6.380638...$ we have $\nu_2^{(+)} (\rho) = \nu_{8}^{(-)} (\rho) \approx
 3.032217...$ For the corresponding eigenfunctions \eqref{eigen}, the amplitude factors
 $U^{(1,2)}_{2,+} (y)$ (solid lines) and $U^{(1,2)}_{8,-} (y)$ (dashed lines) are
 plotted.}
\label{fig:5}
%
\vspace{2.4mm}
%
\centering
 \SetLabels
 \L (0.92*-0.02) $y$\\
 \L (0.895*0.93) {\small $U^{(1)}_{144,-}$}\\
 \L (0.765*0.33) {\small $U^{(2)}_{144,-}$}\\
 \L (0.7*0.57) {\small $U^{(2)}_{3,+}$}\\
  \L (0.905*0.485) {\small $U^{(1)}_{3,+}$}\\
 \endSetLabels
 \leavevmode\AffixLabels{\includegraphics[width=80mm]{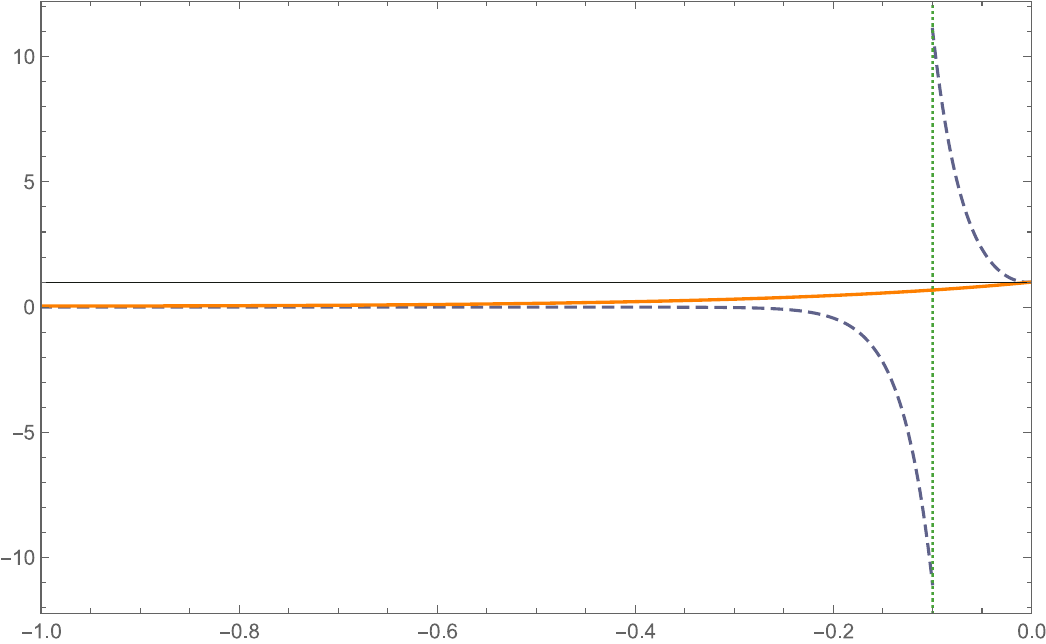}}
 \caption{If $h = 1/10$ in the cylinder with $D = \{ |x| < 1 \}$ and $d = 1$, then for
 $\rho \approx 1.270007...$ we have $\nu_3^{(+)} (\rho) = \nu_{144}^{(-)} (\rho) \approx
 3.827646...$ For the corresponding eigenfunctions \eqref{eigen}, the amplitude factors
 $U^{(1,2)}_{3,+} (y)$ (solid lines) and $U^{(1,2)}_{144,-} (y)$ (dashed lines) are
 plotted.}
\label{fig:6}
\end{figure}

To describe their behaviour, it is convenient to write formulas \eqref{u1} and
\eqref{u2} in the form:
\begin{equation}
u^{(j)}_{m,\pm} (x,y) = \bar{k}_m \, v_m (x) \, U^{(j)}_{m,\pm} (y) .
\label{eigen}
\end{equation}
Here $ j = 1,2$ is the number of layer, $\bar{k}_m$ is the $m$th element of sequence
\eqref{zeroes}, $v_m (x)$ is an eigenfunction of problem \eqref{fm} corresponding to
$\bar{k}_m^2$, and the dependence on $y$ is normalized so that $U^{(1)}_{m, \pm} (0)
= 1$, and expressed as follows:
\begin{gather}
U^{(1)}_{m, \pm} (y) = \cosh (\bar{k}_m y) +
\frac{\nu_m^{(\pm)}}{\bar{k}_m} \sinh (\bar{k}_m y) , \label{Udef} \\ U^{(2)}_{m,
\pm} (y) = \cosh (\bar{k}_m (y+d)) \frac{\nu_m^{(\pm)} \cosh(\bar{k}_m h) -
\bar{k}_m \sinh (\bar{k}_m h)}{\bar{k}_m \sinh (\bar{k}_m (d - h))} \, .
\label{Udef'}
\end{gather}
Lemma~\ref{lemma1} implies that $U^{(2)}_{m,\pm}(-d) \to 0$ as $m \to \infty$.

Thus, for every sloshing eigenfunction all its traces in horizontal cross-sections
of the cylinder have the same graph $v_m (x)$ (see Figs.~\ref{fig:7}--\ref{fig:9}
below), but multiplied by the amplitude factor $U^{(j)}_{m,\pm} (y)$ depending on
the depth in a different way in each of two layers; see formulas \eqref{Udef} and
\eqref{Udef'}. Let us describe several characteristic properties of these amplitude
factors; see also examples of their graphs in Figs.~\ref{fig:4}--\ref{fig:6}. These
examples are chosen so that eigenvalues belonging to different sequences coincide;
see intersections of different lines in Fig.~\ref{fig:2}:
\begin{align*} 
& \mbox{the case} \ \nu_1^{(+)} = \nu_{11}^{(-)} \approx 1.736930... \ \mbox{is
presented in Fig.~\ref{fig:4};} \\ & \mbox{the case} \ \nu_2^{(+)} = \nu_{8}^{(-)}
\approx 3.032217... \ \mbox{is presented in Fig.~\ref{fig:5};} \\ & \mbox{the case}
\ \nu_3^{(+)} = \nu_{144}^{(-)} \approx 3.827646... \ \mbox{is presented in
Fig.~\ref{fig:6}.}
\end{align*}

$\bullet$ All amplitude factors have a jump at $y = -h$, whose size is as follows:
\[ U^{(1)}_{m,\pm}(-h) - U^{(2)}_{m,\pm}(-h) = \frac{\bar{k}_m \sinh(\bar{k}_m d)
- \nu^{(\pm)}_m \cosh(\bar{k}_m d)}{\bar{k}_m \sinh(\bar{k}_m (d - h))} > 0 \ \
\mbox{for all} \ m = 1,2,\dots .
\]
Lemma 1 implies that $U^{(1)}_{m,+}(-h) - U^{(2)}_{m,+}(-h) \to 0$, whereas
$U^{(1)}_{m,-}(-h) - U^{(2)}_{m,-}(-h) \to \infty$ as $m \to \infty$.

$\bullet$ It follows by differentiation that both $U^{(1)}_{m,+}$ and
$U^{(2)}_{m,+}$ decrease monotonically with depth, remaining positive in view of
Proposition~\ref{prop4}, but attaining rather small values at the bottom; in particular:
\begin{align*} 
& U^{(2)}_{1,+}(-d) \approx 0.306440... \ \mbox{in Fig.~\ref{fig:4};} \\ &
U^{(2)}_{2,+}(-d) \approx 0.093735... \ \mbox{in Fig.~\ref{fig:5};} \\ &
U^{(2)}_{3,+}(-d) \approx 0.043316... \ \mbox{in Fig.~\ref{fig:6}.}
\end{align*}
Moreover, the corresponding jumps at $y = -h$ are also small; in particular:
\begin{align*} 
& U^{(1)}_{1,+}(-h) - U^{(2)}_{1,+}(-h) \approx 0.009632... \ \mbox{in
Fig.~\ref{fig:4};} \\ & U^{(1)}_{2,+}(-h) - U^{(2)}_{2,+}(-h) \approx 0.003785... \
\mbox{in Fig.~\ref{fig:5};} \\ & U^{(1)}_{3,+}(-h) - U^{(2)}_{3,+}(-h) \approx
0.000177... \ \mbox{in Fig.~\ref{fig:6}.}
\end{align*}

$\bullet$ The function $U^{(2)}_{m,-}$ increases monotonically with depth remaining
negative in view of Proposition~\ref{prop4}, and approaching very close to zero at
the bottom; in particular:
\begin{align*} 
& U^{(2)}_{11,-}(-d) \approx -0.000883... \ \mbox{in Fig.~\ref{fig:4};} \\ &
U^{(2)}_{8,-}(-d) \approx -0.000786... \ \mbox{in Fig.~\ref{fig:5};} \\ &
U^{(2)}_{144,-}(-d) \approx - (5.545769...) \cdot 10^{-12} \ \mbox{in
Fig.~\ref{fig:6},}
\end{align*}
which agrees with what was said after formula \eqref{Udef'}.

$\bullet$ The function $U^{(1)}_{m,-}$ is positive, but nonmonotonic on $(-h, 0)$.
Indeed, \eqref{Udef} implies that $\bigl[ U^{(1)}_{m,-} \bigr]' (y)$ vanishes only at
\[ y =  y_m^* = (2 \bar{k}_m)^{-1} \log \bigl[ \bigl( \bar{k}_m - \nu^{(-)}_m \bigr) /
\bigl( \bar{k}_m + \nu^{(-)}_m \bigr) \bigr] \, .
\]
Moreover, since $\bar{k}_m > \nu^{(-)}_m$ (see Proposition~\ref{prop4}), we have
\[ \bigl[ U^{(1)}_{m,-} \bigr]'' (y_*) = \bar{k}_m  \sqrt{\bar{k}_m^2 - \bigl[
\nu^{(-)}_m \bigr]^2} > 0 \, ,
\]
and so $U^{(1)}_{m,-} (y)$ attains its minimum at $y =  y_m^*$. These points and the
corresponding values of minima (they lie on thin horizontal lines in
Figs.~\ref{fig:4}--\ref{fig:6}) are as follows in our examples:
\begin{align*} 
& y_{11}^* \approx -0.027472... \, , \quad U^{(1)}_{11,-} (y_{11}^*) \approx
0.976237... \ \mbox{in Fig.~\ref{fig:4};} \\ & y_8^* \approx -0.072680... \, , \quad
U^{(1)}_{8,-} (y_8^*) \approx 0.891938... \ \mbox{in Fig.~\ref{fig:5};} \\ &
y_{144}^* \approx -0.003700... \, , \quad U^{(1)}_{144,-} (y_{144}^*) \approx
0.992926... \ \mbox{in Fig.~\ref{fig:6}.}
\end{align*}

$\bullet$ Unlike $U^{(1)}_{m,+}(-h) - U^{(2)}_{m,+}(-h)$, which is small for every
$m = 1,2,\dots$, the analogous jump with ``$-$'' instead of ``$+$'' has a
significant value as follows from our examples:
\begin{align*} 
& U^{(1)}_{11,-}(-h) - U^{(2)}_{11,-}(-h) \approx 1.745916... \ \mbox{in
Fig.~\ref{fig:4};} \\ & U^{(1)}_{8,-}(-h) - U^{(2)}_{8,-}(-h) \approx 1.071275... \
\mbox{in Fig.~\ref{fig:5};} \\ & U^{(1)}_{144,-}(-h) - U^{(2)}_{144,-}(-h) \approx
22.138106... \ \mbox{in Fig.~\ref{fig:6}.}
\end{align*}

\begin{figure}[t!]
\centering
 \SetLabels
 \L (0.08*0.14) $x_1$\\
 \L (0.61*0.14) $x_1$\\
 \L (0.345*0.08) $x_2$\\
 \L (0.875*0.08) $x_2$\\
 \L (-0.02*0.58) $v_1$\\
 \L (0.50*0.58) $v_{11}$\\
  \endSetLabels
 \leavevmode\mbox{\kern2.5mm}\AffixLabels{\includegraphics[width=63mm]{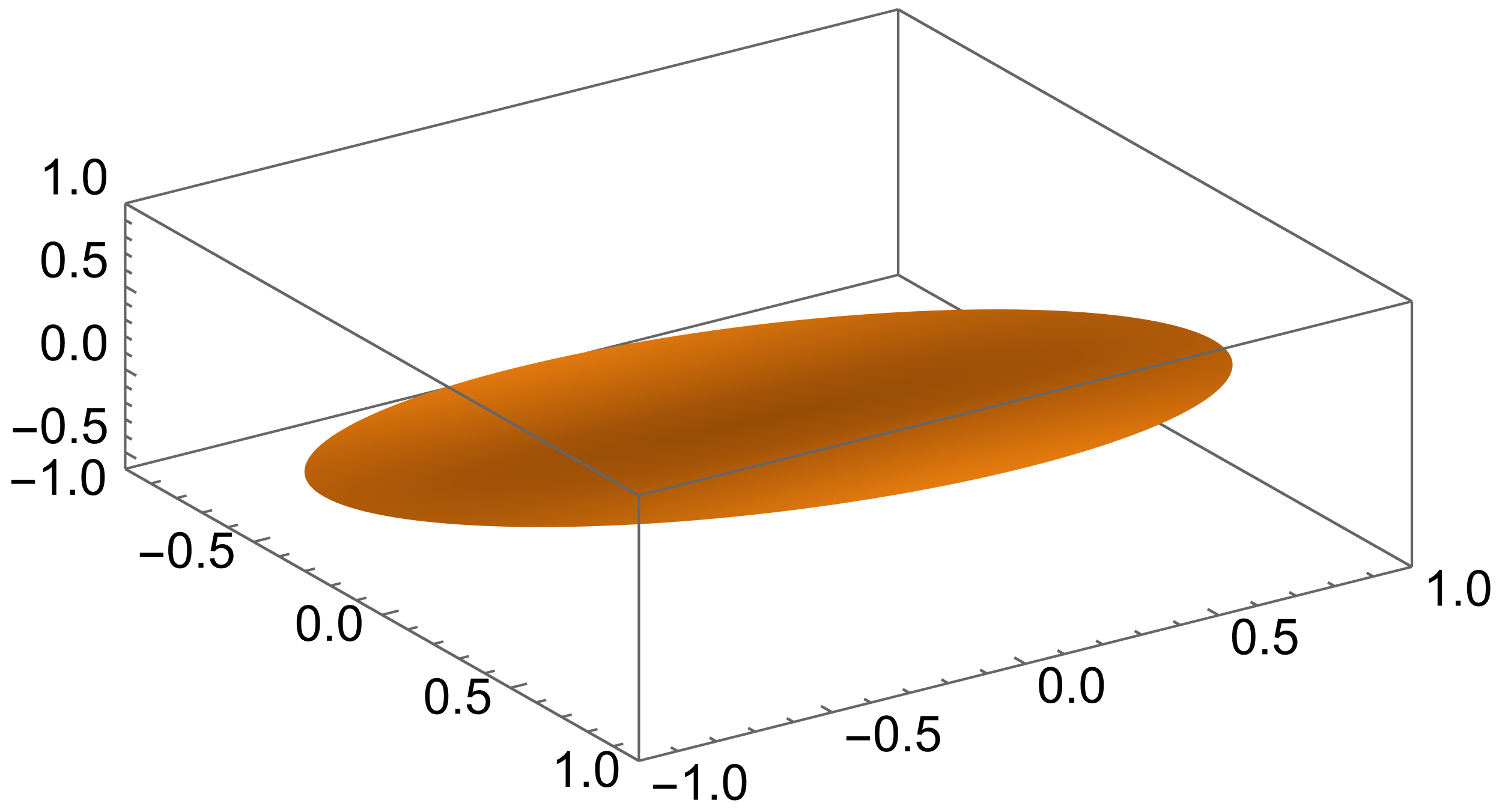}
 \kern8mm\includegraphics[width=65mm]{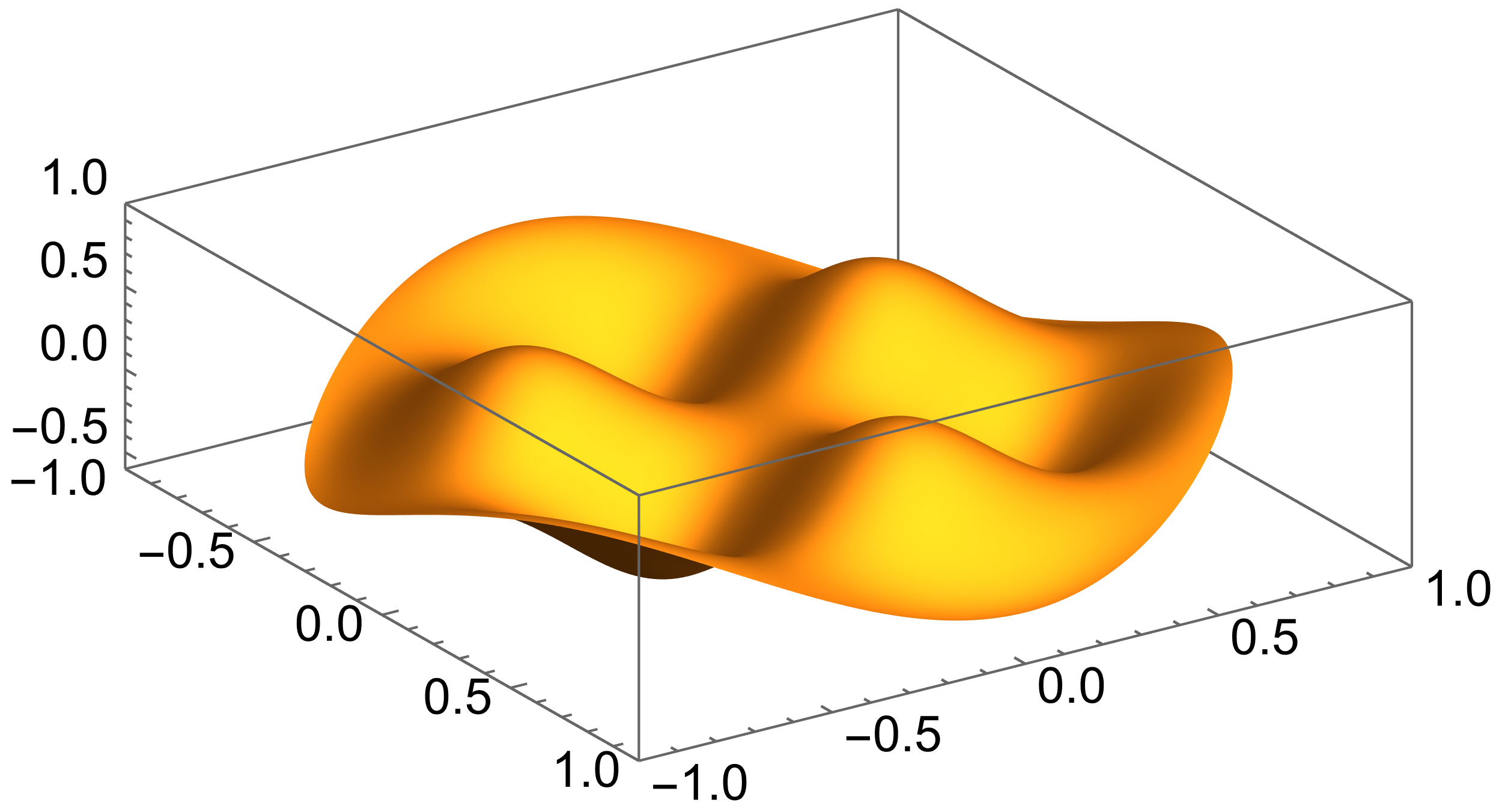}}
 \caption{If $h = 1/10$ in the cylinder with $D = \{ |x| < 1 \}$ and $d = 1$, then for
 $\rho \approx 1.804998...$ we have $\nu_1^{(+)} (\rho) = \nu_{11}^{(-)} (\rho) \approx
 1.736930...$ The graphs of traces on $D$ are plotted for the corresponding
 eigenfunctions \eqref{eigen}: $v_1 (x)$ (left) and $v_{11} (x)$ (right).}
  \label{fig:7}
%
\vspace{6mm}
%
\centering
 \SetLabels
 \L (0.08*0.14) $x_1$\\
 \L (0.61*0.14) $x_1$\\
 \L (0.345*0.08) $x_2$\\
 \L (0.875*0.08) $x_2$\\
 \L (-0.02*0.58) $v_2$\\
 \L (0.508*0.58) $v_8$\\
  \endSetLabels
 \leavevmode\mbox{\kern2.5mm}\AffixLabels{\includegraphics[width=63mm]{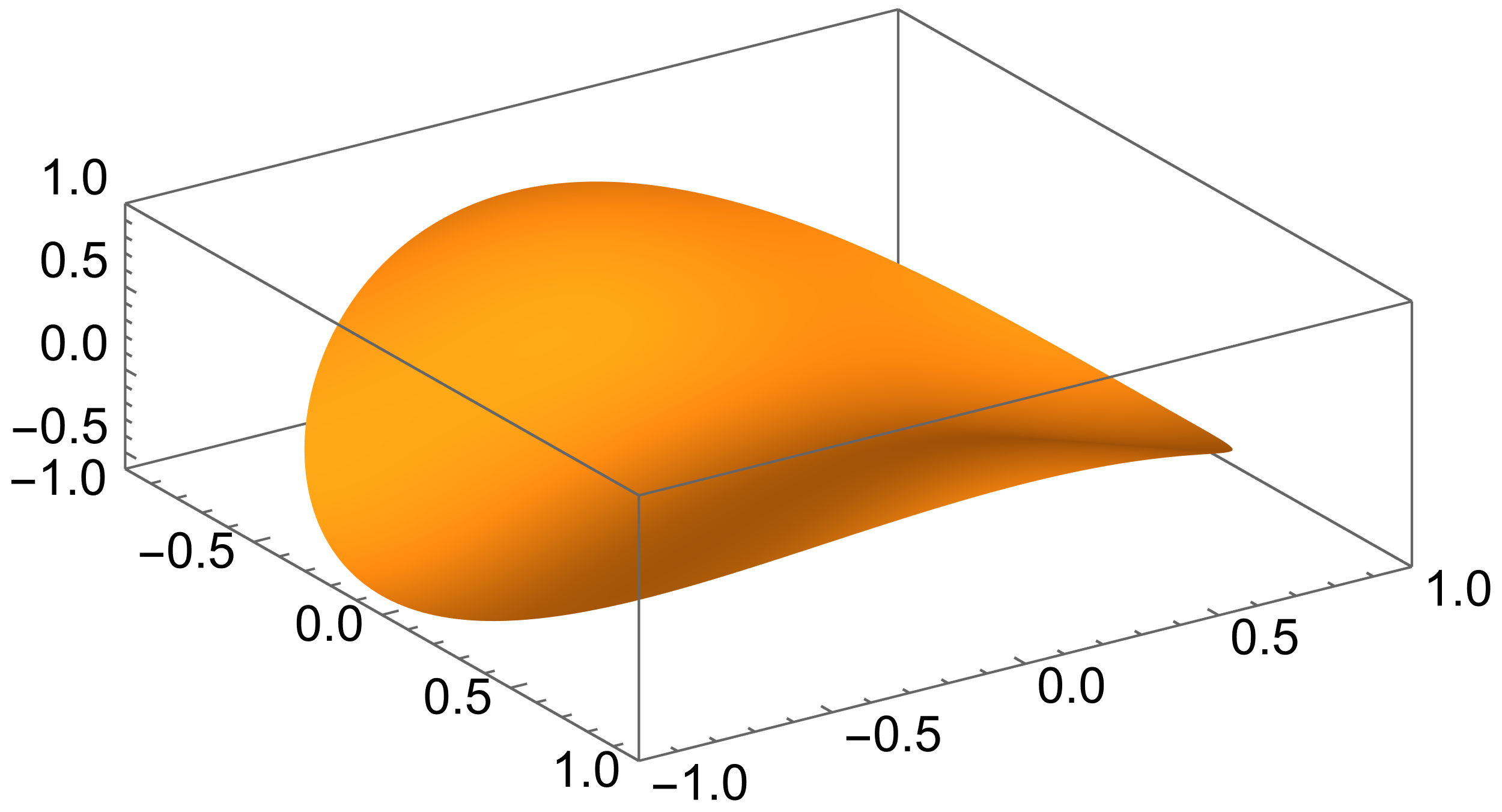}
 \kern8mm\includegraphics[width=65mm]{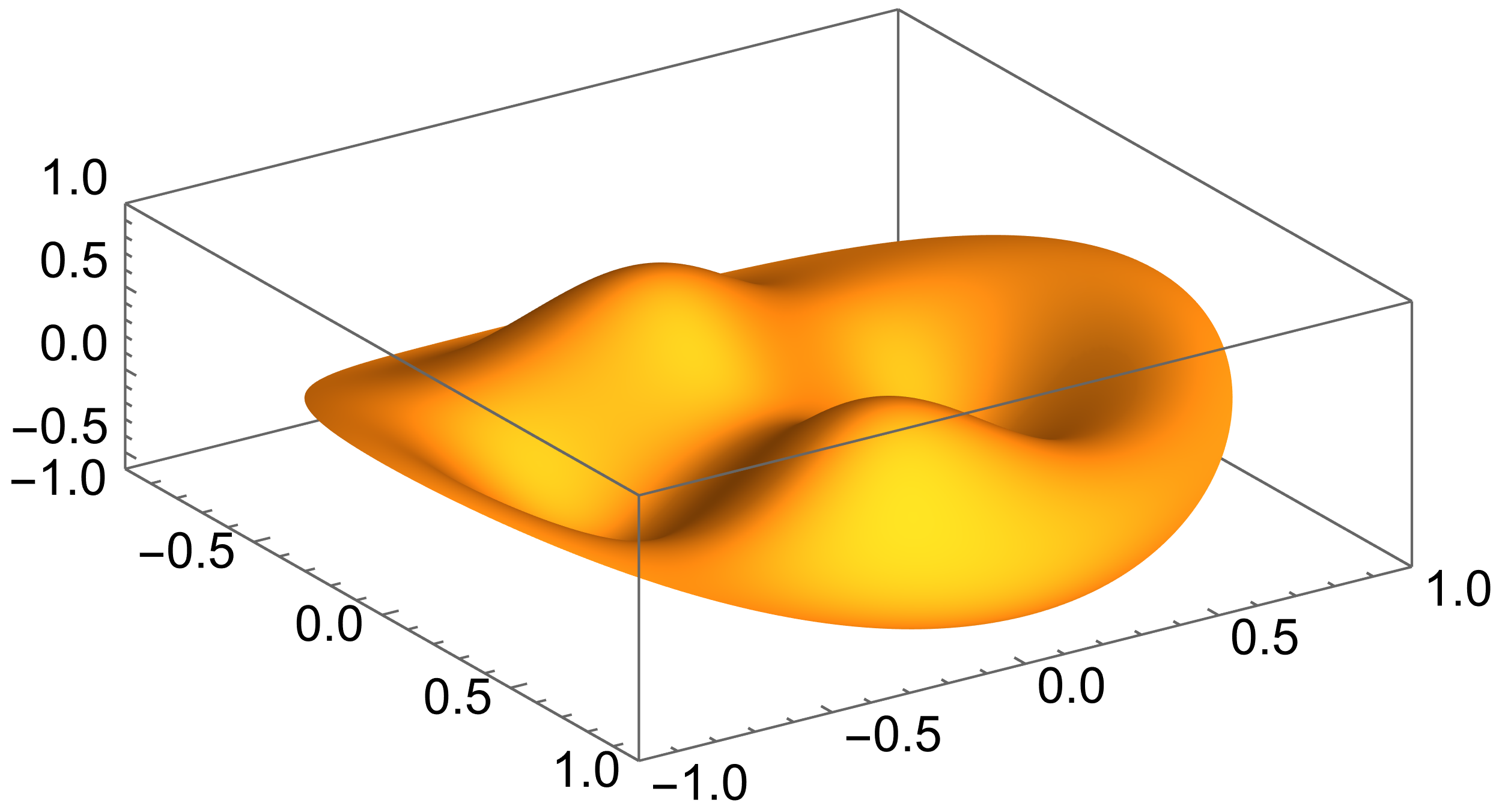}}
 \caption{If $h = 1/10$ in the cylinder with $D = \{ |x| < 1 \}$ and $d = 1$, then for
 $\rho \approx 6.380638...$ we have $\nu_2^{(+)} (\rho) = \nu_{8}^{(-)} (\rho) \approx
 3.032217...$ The graphs of traces on $D$ are plotted for the corresponding
 eigenfunctions \eqref{eigen}: $v_2(x)$ (left) and $v_8(x)$ (right).}
\label{fig:8}
%
\vspace{6mm}
%
\centering
 \SetLabels
 \L (0.08*0.14) $x_1$\\
 \L (0.61*0.14) $x_1$\\
 \L (0.345*0.08) $x_2$\\
 \L (0.875*0.08) $x_2$\\
 \L (-0.02*0.58) $v_3$\\
 \L (0.49*0.58) $v_{144}$\\
  \endSetLabels
 \leavevmode\mbox{\kern2.5mm}\AffixLabels{\includegraphics[width=64mm]{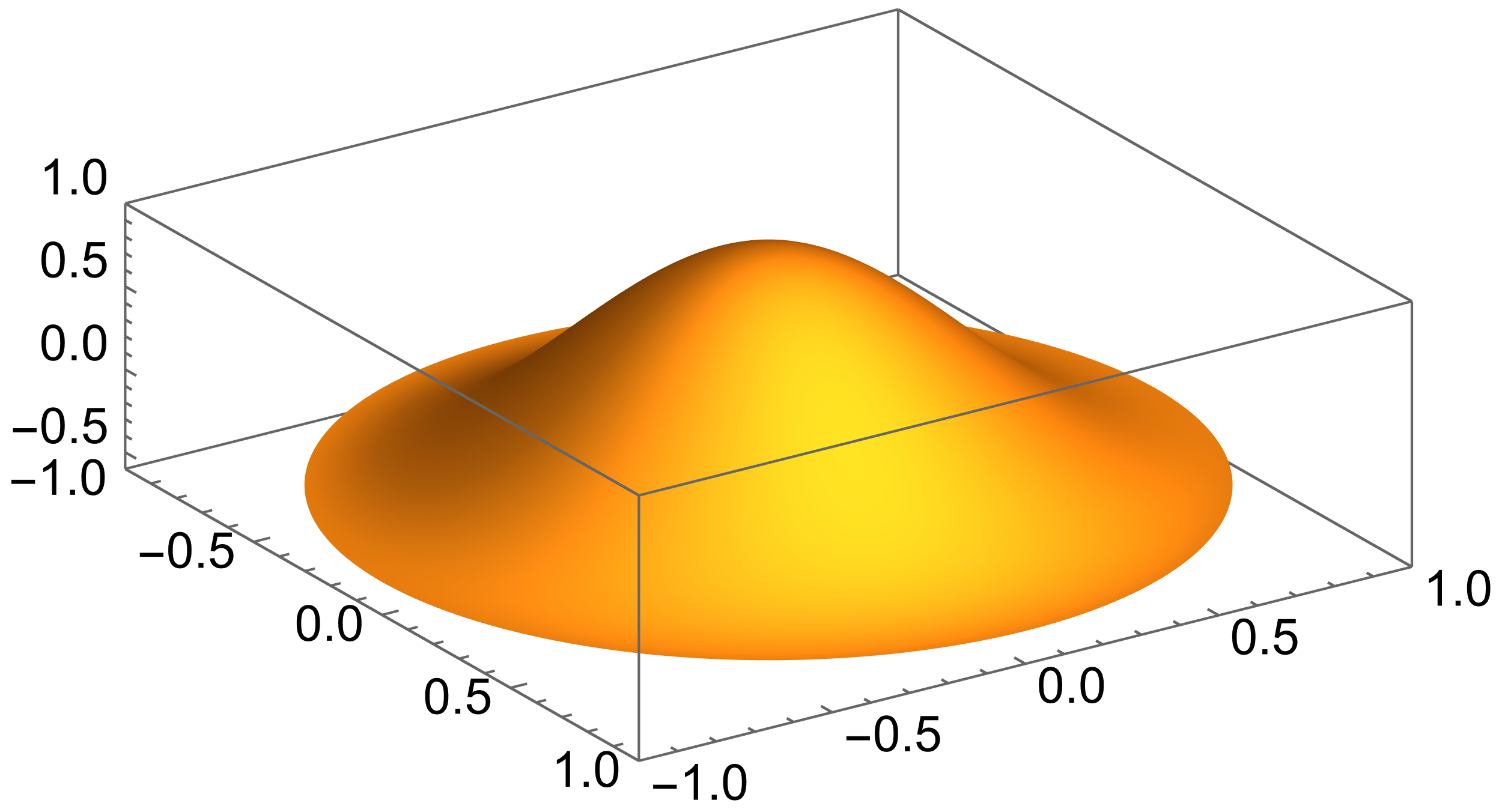}
 \kern8mm\includegraphics[width=64mm]{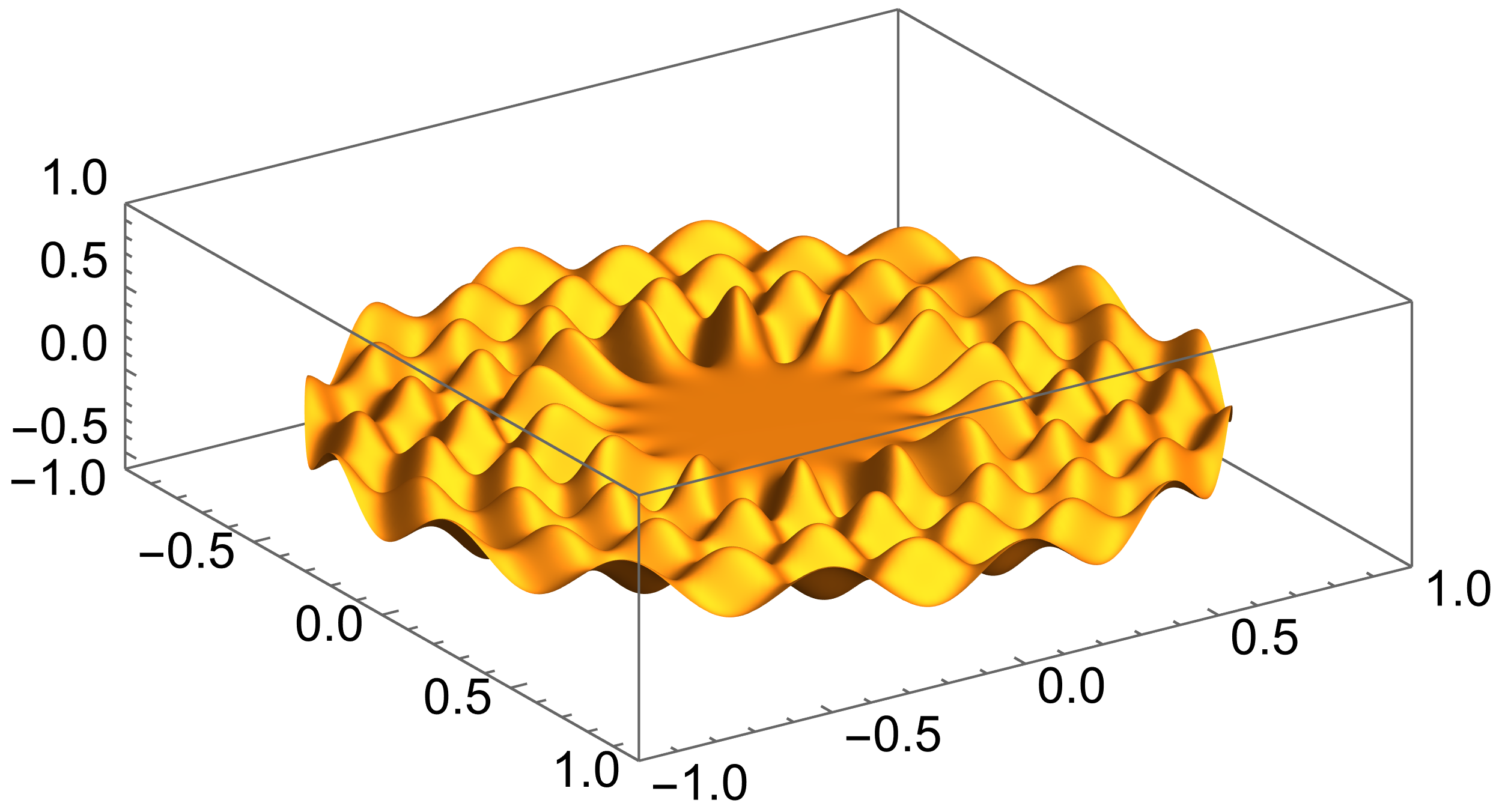}}
 \caption{If $h = 1/10$ in the cylinder with $D = \{ |x| < 1 \}$ and $d = 1$, then for
 $\rho \approx 1.270007...$ we have $\nu_3^{(+)} (\rho) = \nu_{144}^{(-)} (\rho) \approx
 3.827646...$ The graphs of traces on $D$ are plotted for the corresponding
 eigenfunctions \eqref{eigen}: $v_3(x)$ (left) and $v_{144}(x)$ (right).}
  \label{fig:9}
\end{figure}

Now, we illustrate the dependence of sloshing eigenfunctions \eqref{eigen} on $x$;
namely, the graphs of their traces on the cylinder's horizontal cross-section $D$
are shown in Figs.~\ref{fig:7}--\ref{fig:9}. When the multiplicity of
$\nu_m^{(\pm)}$ is two (this happens when $\bar{k}_m = j'_{p, s}$ and $p \neq 0$),
the eigenmode $v_m (x) = J_p (j'_{p,s} |x|) (x_1 / |x|)$ is plotted.

It is worth emphasizing the contrast between the left and right graphs in
Fig.~\ref{fig:9}, because these eigenmodes correspond to the same eigenvalue
$\approx 3.827646...$, but appearing in different sequences. The contrast becomes
clear when one takes into account the large difference in numbers of this eigenvalue
in each of these sequences. On the other hand, if the difference in numbering is not
so large, then there is no such a contrast in eigenmodes corresponding to the same
eigenvalue; see Figs.~\ref{fig:7} and \ref{fig:8}.




\section{The effect of a radial baffle in a circular cylinder}
\label{sect:3}

In the recent paper \cite{KM}, the behaviour of sloshing eigenvalues was studied in
the case when a vertical circular cylinder of finite depth is occupied by a
homogeneous fluid. The effect of breaking the axial symmetry due to a radial baffle
was analysed; the baffle was assumed to extend throughout the fluid depth. It was
demonstrated that all eigenvalues are simple in the presence of baffle, whereas the
lowest eigenvalue and many others have multiplicity two in its absence. These
results follow from properties of spectra of two problems; one of them is problem
\eqref{fm} in a disk, whereas the other one is a similar problem in a disk with a
radial cut from the centre to the boundary. Any other parameter (cylinder's depth,
etc.) is unimportant for these results.

It is demonstrated above (see Figs.~\ref{fig:2} and \ref{fig:3}) that sloshing
eigenvalues have multiplicity depending on $\rho$ when a two-layer fluid occupies a
circular cylinder without a radial baffle. The same occurs in the case when a radial
baffle is present; see Fig.~\ref{fig:10}, where the case of cylinder, whose
cross-section is $D = \{ |x|
 < 1 \} \setminus \{ x_1 \geq 0, \ x_2 = 0 \}$ (the unit disk with a radial cut),
$d = 1$ and $h = 1/10$ is illustrated; namely, the graphs of simple eigenvalues
$\nu_m^{(+)} (\rho)$ (bold lines) and $\nu_m^{(-)} (\rho)$ (see formula \eqref{nupm}
with $k = k_m$) are plotted for $m = 1,2,\dots,7$ and $m = 1,2,\dots,21,281$,
respectively; the vertical lines at $\rho = 1.3$ and $\rho = 13.5$ correspond to the
gasoline/water and water/mercury superpositions, respectively. Hence the
multiplicity is two at every point of intersection. 

\begin{figure}[t!]
\centering
 \SetLabels
 \L (0.956*-0.02) $\rho$\\
 \L (0.035*0.708) {\small{}$\nu^{(-)}_{280,...,22}$}\\
 \endSetLabels
 \leavevmode\AffixLabels{\includegraphics[width=125mm]{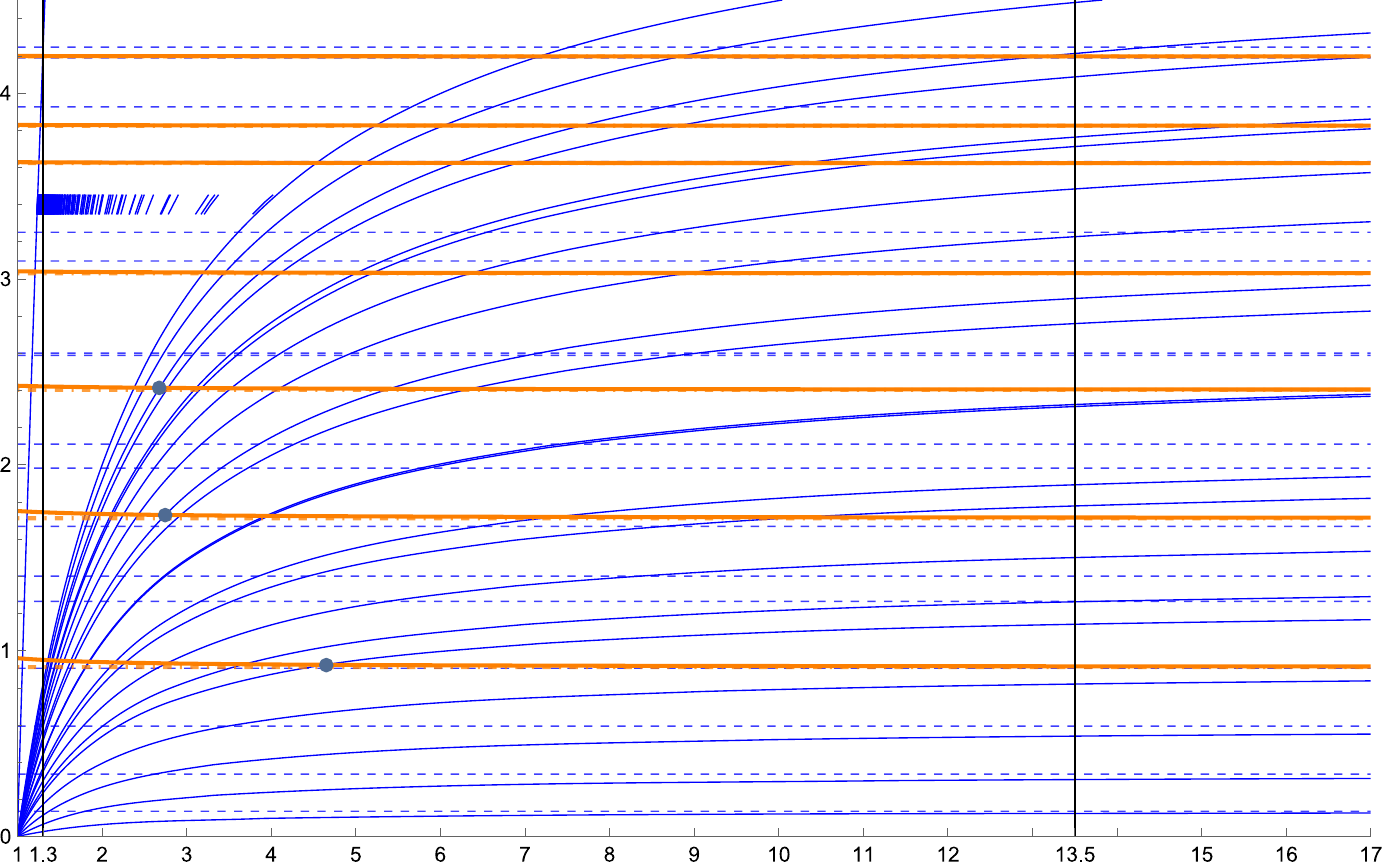}}
 \vspace{2mm}
 \caption{For the cylinder with $d = 1$, $h = 1/10$ and $D = \{ |x| < 1 \} \setminus \{
 x_1 \geq 0, \ x_2 = 0 \}$ (the unit disk with the radial cut along the $x_1$-axis)
 the graphs $\nu^{(+)}_{1,...,7} (\rho)$ (bold lines) and $\nu^{(-)}_{1,...,21,281}
 (\rho)$ are plotted for $\rho \in [0, 17]$, whereas the positions of
 $\nu^{(-)}_{22,...,280} (\rho)$ are just located between $\nu^{(+)}_4$ and
 $\nu^{(+)}_5$. The corresponding asymptotes are shown by dashed lines. Bullets mark
 eigenvalues for which properties of the corresponding eigenfunctions are investigated.}
  \label{fig:10}
\end{figure}

Here the initial eigenvalues of problem \eqref{fm} are defined by
\begin{align*}
&k_1 = j'_{1/2,1} \approx 1.165561...,&& k_2 = j'_{1,1} \approx 1.841183...,&& k_3
= j'_{3/2,1} \approx 2.460535...,\\ &k_4 = j'_{2,1} \approx 3.054236...,&& k_5 =
j'_{5/2,1} \approx 3.632797...,&& k_6 = j'_{0,1} \approx 3.831705...,\\
 &k_7 = j'_{3,1} \approx 4.201188...,&& k_8 = j'_{1/2,2} \approx 4.604216...,&& k_9
= j'_{7/2,1} \approx 4.762196...,\\ &k_{10} = j'_{4,1} \approx 5.317553...,&&
 k_{11} = j'_{1,2} \approx 5.331442...,&& k_{12} = j'_{9/2,1} \approx 5.868419...,\\
&k_{13} = j'_{3/2,2} \approx 6.029292...,&& k_{14} = j'_{5,1} \approx
6.415616...,&& k_{15} = j'_{2,2} \approx 6.706133..., \\ &k_{16} = j'_{11/2,1}
\approx 6.959745...,&& k_{17} = j'_{0,2} \approx 7.015586...,&& k_{18} = j'_{5/2,2}
\approx 7.367008...,\\ &k_{19} = j'_{6,1} \approx 7.501266...,&& k_{20} = j'_{1/2,3}
\approx 7.789883...,&& k_{21} = j'_{3,2} \approx 8.015236...,
\end{align*}
whereas, $k_{280} = j'_{0,10} \approx 32.189679...$, $k_{281} = j'_{14,5} \approx
32.236969...$; recall that $j'_{\mu,s}$ is the $s$th positive zero of $J'_\mu$ (for
$\mu = 0$ this numbering differs from that used in \cite{AS}, where $j'_{0,s}$ is
the $s$th nonnegative zero).

We see that below the level~4 there are twice as many solid lines in
Fig.~\ref{fig:10} as in Fig.~\ref{fig:2}, whereas the number of eigenvalues
$\nu_m^{(-)} (1.3)$ below this level is also almost two times larger in
Fig.~\ref{fig:10} comparing with Fig.~\ref{fig:2}; the sequence in Fig.~\ref{fig:10}
for $\rho=1.3$ is as follows:
\begin{multline*}
\bigl\{ \nu^{(-)}_{1}, \dots, \nu^{(-)}_{25}, \nu^{(+)}_{1}, \nu^{(-)}_{26}, \dots,
\nu^{(-)}_{60}, \nu^{(+)}_{2}, \nu^{(-)}_{61}, \dots, \nu^{(-)}_{102},
\nu^{(+)}_{3}, \nu^{(-)}_{103}, \dots, \nu^{(-)}_{153}, \nu^{(+)}_{4}, \\
\nu^{(-)}_{154}, \dots, \nu^{(-)}_{213}, \nu^{(+)}_{5}, \nu^{(-)}_{214}, \dots,
\nu^{(-)}_{235}, \nu^{(+)}_{6}, \nu^{(-)}_{236}, \dots \nu^{(-)}_{281},
\nu^{(+)}_{7}, \nu^{(-)}_{282}, \dots \bigr\} \, , 
\end{multline*}
where
\begin{align*}
&\nu^{(-)}_{1} \approx 0.027974...,&& \nu^{(-)}_{25} \approx 0.907722...,&&
\nu^{(+)}_{1} \approx 0.949859...,&& \nu^{(-)}_{26} \approx 0.976046...,\\
&\nu^{(-)}_{60} \approx 1.718432...,&& \nu^{(+)}_{2} \approx 1.744132...,&&
\nu^{(-)}_{61} \approx 1.760512...,&& \nu^{(-)}_{102} \approx 2.416412...,\\
&\nu^{(+)}_{3} \approx 2.421447...,&& \nu^{(-)}_{103} \approx 2.430348...,&&
\nu^{(-)}_{153} \approx 3.021836...,&& \nu^{(+)}_{4} \approx 3.039090...,\\
&\nu^{(-)}_{154} \approx 3.039249...,&& \nu^{(-)}_{213} \approx 3.623983...,&&
\nu^{(+)}_{5} \approx 3.627035...,&& \nu^{(-)}_{214} \approx 3.627614...,\\
&\nu^{(-)}_{235} \approx 3.818197...,&& \nu^{(+)}_{6} \approx 3.827601...,&&
\nu^{(-)}_{236} \approx 3.829050...,&& \nu^{(-)}_{281} \approx 4.197291...,\\
&\nu^{(+)}_{7} \approx 4.199020...,&& \nu^{(-)}_{282} \approx 4.207027...
\end{align*}
Notice that when a homogeneous fluid occupies the same container just the eigenvalue
$\nu_8$ has the value close to $\nu^{(-)}_{282}$.

For this geometry, there is no asymptotic formula for the counting function of
eigenvalues similar to obtained in Proposition~\ref{prop5}. Indeed, Weyl's law is
not valid for cut disks because of singularities that eigenfunctions have at the
cut's tip.

Furthermore, it is remarkable that the value of $\nu^{(-)}_1 (1.3)$ in the presence
of baffle is more than 2.4 times smaller than that in its absence. It is worth
reminding (see \cite{KM}), that in a homogeneous fluid with a baffle, the value of
$\nu_1$ for $d=1$ is about $1.8$ times smaller than that in the absence; see
Fig.~\ref{fig:ratio'}. This demonstrates that the effect of a baffle manifests
itself more strongly in a two-layer fluid.

Let us describe the effect of a baffle for arbitrary values of $d$, $h$ and $\rho$.
In order to distinguish $\nu^{(-)}_1 (\rho,h,d)$ from the analogous eigenvalue in
the absence of baffle, we denote the latter by $\bar \nu^{(-)}_1 (\rho,h,d)$. We
define the magnification factor (it compares the case of a two-layer fluid with that
of a homogeneous one) as follows:
\[ K (\rho,h,d) =
\frac{\nu^{(-)}_1 (\rho,h,d)}{\bar \nu^{(-)}_1 (\rho,h,d)} \biggm/
\frac{k_1 \tanh (k_1 d)}{\vphantom{\nu^{(-)}_1}\bar{k}_1 \tanh (\bar{k}_1 d)} \, .
\]
In both ratios the fundamental sloshing eigenvalue in a circular cylinder without
the radial baffle stands in the denominator. It follows from numerical computations
that
\[ K (\rho,h,d) \geq K (1,h,d) \geq 1 .
\]
The dependence of $K (1,h,d)$ (it provides the lower bound for the magnification
factor) on $h/d \in (0, 1/2]$ and $d \in (0, 5]$ is shown in Fig.~\ref{fig:ratio}.

\begin{figure}[t!]
\centering
 \SetLabels
 \L (0.92*-0.04) {\small$d$}\\
 \L (-0.26*0.78) {\small$\displaystyle\frac{\bar{k}_1 \tanh (\bar{k}_1 d)}{k_1 \tanh (k_1
 d)}$}\\
 \endSetLabels
 \leavevmode\AffixLabels{\includegraphics[width=75mm]{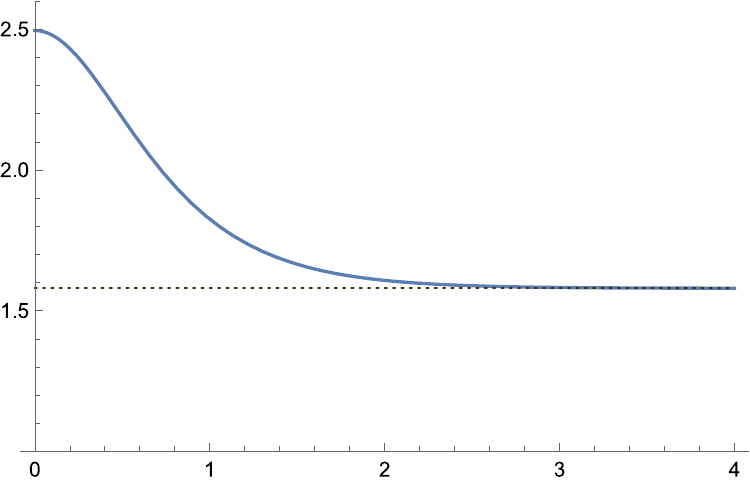}}
 \caption{The graph of  the ratio of the fundamental sloshing eigenvalues for a
 circular cylinder without (numerator) and with the radial baffle is plotted for $d
 \in (0, 4]$. It monotonically decreases from the limit value $\bar{k}_1^2 / k_1^2
 \approx 2.495307...$ as $d \to 0$ and asymptotes the value $\bar{k}_1 / k_1 \approx
 1.579654...$ as $d \to \infty$.}
 \label{fig:ratio'}
\end{figure}

\begin{figure}
\centering
 \SetLabels
 \L (-0.022*0.585) $K$\\
 \L (0.68*0.06) $d$\\
 \L (0.14*0.08) $h/d$\\
 \endSetLabels
 \leavevmode\AffixLabels{\includegraphics[width=75mm]{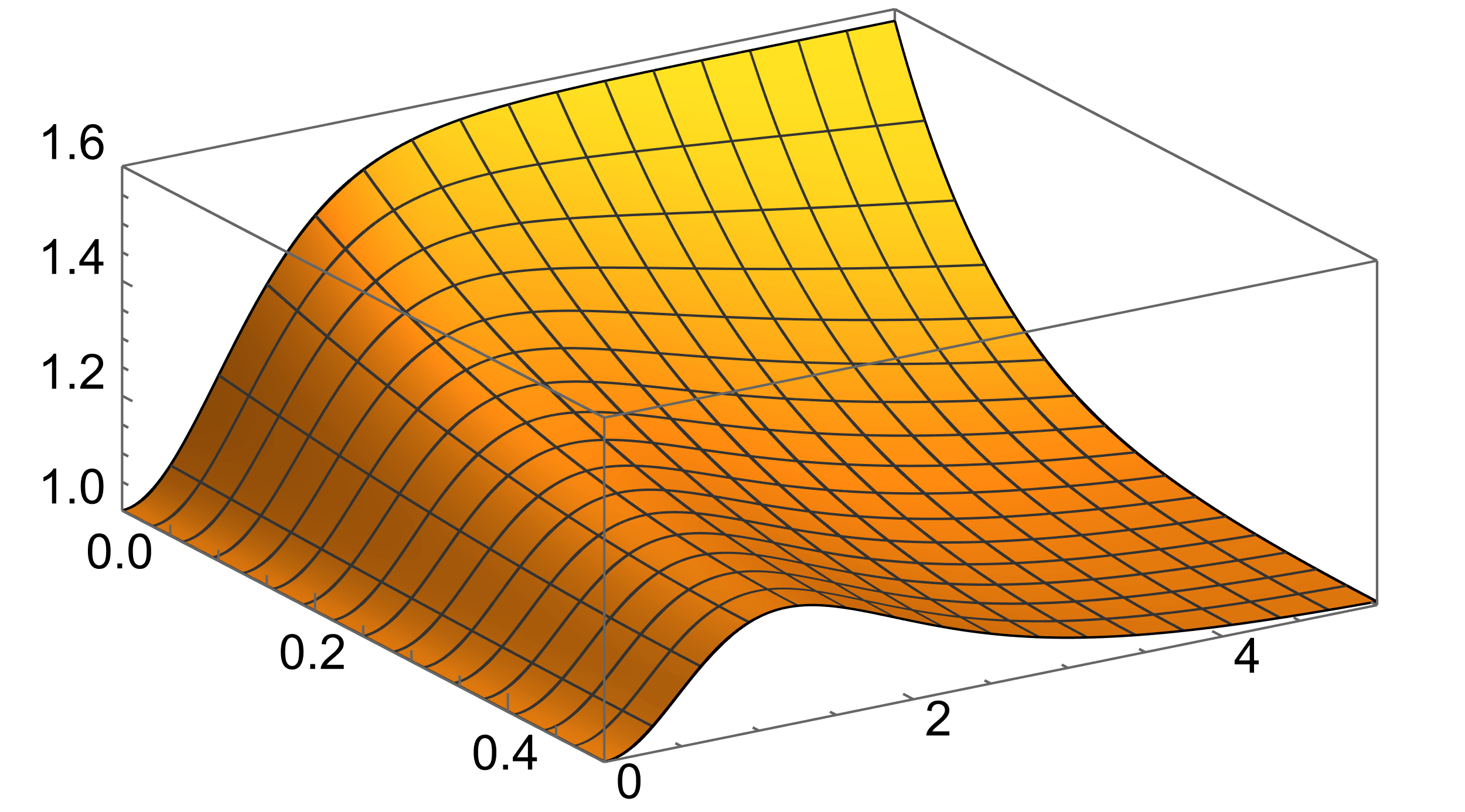}}
 \vspace{-2mm}
 \caption{The graph of $K (1,h,d)$\,---\,the lower bound for the magnification factor
 which increases the ratio $k_1\tanh k_1 d / \bar{k}_1\tanh \bar{k}_1 d$ (see
 Fig.~\ref{fig:ratio'}, where the graph of its reciprocal is plotted) in a two-layer
 fluid\,---\,is plotted for $h/d \in (0, 1/2]$ and $d \in (0, 5]$.}
 \label{fig:ratio}
\end{figure}

\begin{figure}[p]
\centering
 \SetLabels
 \L (0.92*-0.02) $y$\\
 \L (0.905*0.97) {\small $U^{(1)}_{5,-}$}\\
 \L (0.57*0.17) {\small $U^{(2)}_{5,-}$}\\
 \L (0.46*0.71) {\small $U^{(2)}_{1,+}$}\\
  \L (0.905*0.85) {\small $U^{(1)}_{1,+}$}\\
 \endSetLabels
 \leavevmode\AffixLabels{\includegraphics[width=80mm]{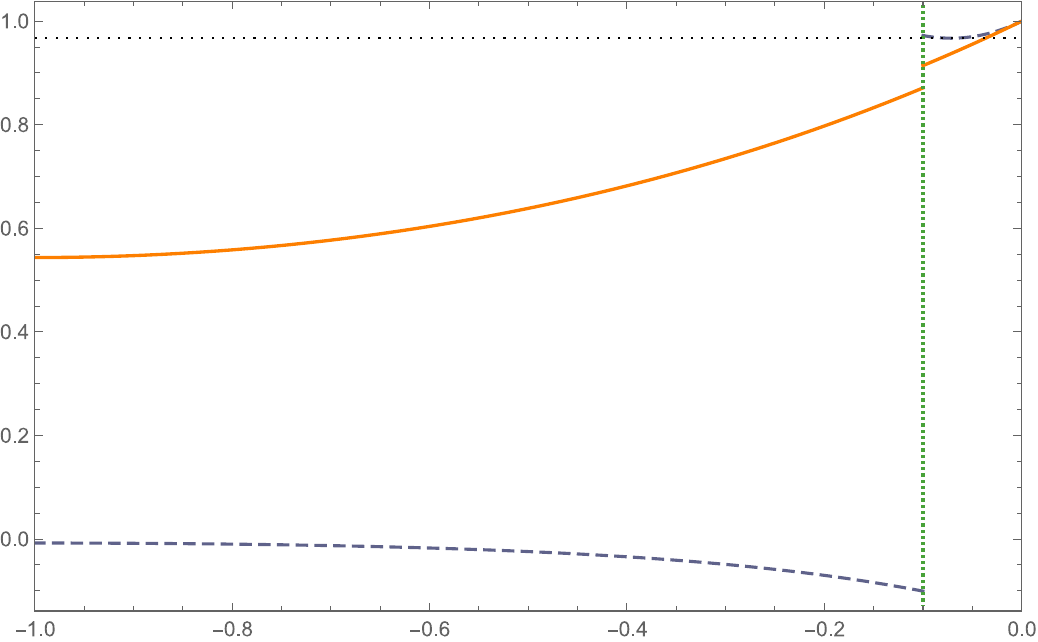}}
 \caption{If $h = 1/10$ in the cylinder with $D = \{ |x| < 1 \} \setminus \{ x_1 \geq 0,
 \ x_2 = 0 \}$ and $d = 1$, then for $\rho \approx 4.651901...$ we have $\nu_1^{(+)}
 (\rho) = \nu_{5}^{(-)} (\rho) \approx 0.923307...$ The amplitude factors
 $U^{(1,2)}_{1,+} (y)$ (solid lines) and  $U^{(1,2)}_{5,-} (y)$ (dashed lines) of
 eigenfunctions \eqref{eigen} are plotted.}
\label{fig:11}
%
\vspace{2.4mm}
%
\centering
 \SetLabels
 \L (0.92*-0.02) $y$\\
 \L (0.895*0.95) {\small $U^{(1)}_{13,-}$}\\
 \L (0.47*0.31) {\small $U^{(2)}_{13,-}$}\\
 \L (0.55*0.67) {\small $U^{(2)}_{2,+}$}\\
  \L (0.905*0.79) {\small $U^{(1)}_{2,+}$}\\
 \endSetLabels
 \leavevmode\AffixLabels{\includegraphics[width=80mm]{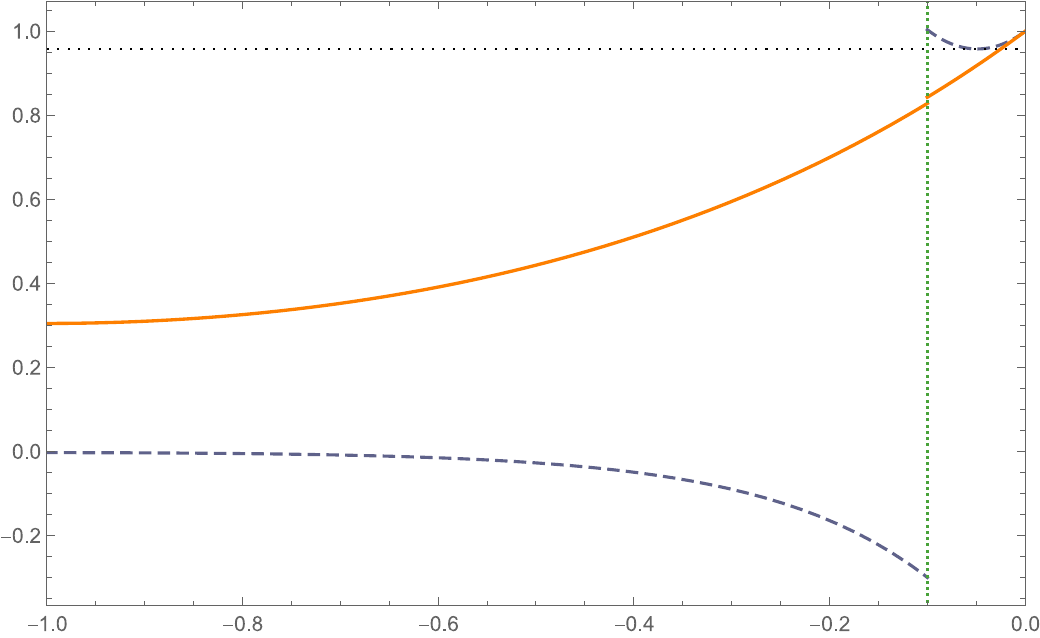}}
 \caption{If $h = 1/10$ in the cylinder with $D = \{ |x| < 1 \} \setminus \{ x_1 \geq 0,
 \ x_2 = 0 \}$ and $d = 1$ and $d = 1$, then for $\rho \approx 2.750213...$ we have
 $\nu_2^{(+)} (\rho) = \nu_{13}^{(-)} (\rho) \approx 1.729565...$ The amplitude factors
 $U^{(1,2)}_{2,+} (y)$ (solid lines) and $U^{(1,2)}_{13,-} (y)$ (dashed lines) of
 eigenfunctions \eqref{eigen} are plotted.}
\label{fig:12}
%
\vspace{2.4mm}
%
\centering
 \SetLabels
 \L (0.92*-0.02) $y$\\
 \L (0.895*0.95) {\small $U^{(1)}_{19,-}$}\\
 \L (0.65*0.21) {\small $U^{(2)}_{19,-}$}\\
 \L (0.6*0.635) {\small $U^{(2)}_{3,+}$}\\
  \L (0.905*0.78) {\small $U^{(1)}_{3,+}$}\\
 \endSetLabels
 \leavevmode\AffixLabels{\includegraphics[width=80mm]{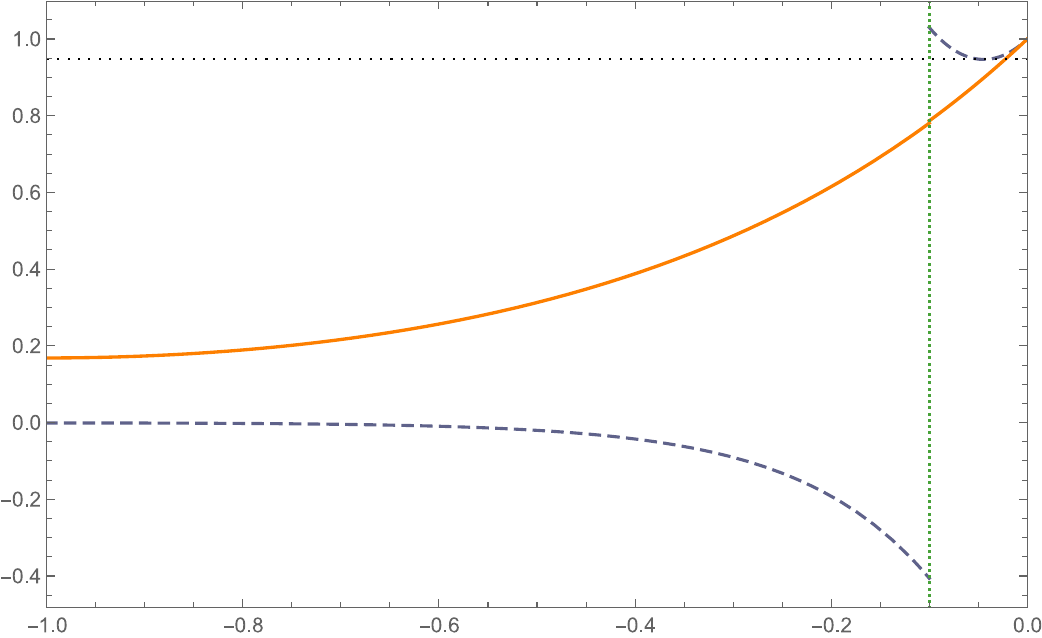}}
 \caption{If $h = 1/10$ in the cylinder with $D = \{ |x| < 1 \} \setminus \{ x_1 \geq 0,
 \ x_2 = 0 \}$ and $d = 1$, then for $\rho \approx  2.678476...$ we have $\nu_3^{(+)}
 (\rho) = \nu_{19}^{(-)} (\rho) \approx 2.413589...$ The amplitude factors
 $U^{(1,2)}_{3,+} (y)$ (solid lines) and $U^{(1,2)}_{19,-} (y)$ (dashed lines) of
 eigenfunctions \eqref{eigen} are plotted.}
\label{fig:13}
\end{figure}

\begin{figure}[t!]
\centering
 \SetLabels
 \L (0.08*0.14) $x_1$\\
 \L (0.61*0.14) $x_1$\\
 \L (0.345*0.08) $x_2$\\
 \L (0.875*0.08) $x_2$\\
 \L (-0.02*0.58) $v_1$\\
 \L (0.50*0.58) $v_{5}$\\
  \endSetLabels
 \leavevmode\mbox{\kern2.5mm}\AffixLabels{\includegraphics[width=63mm]{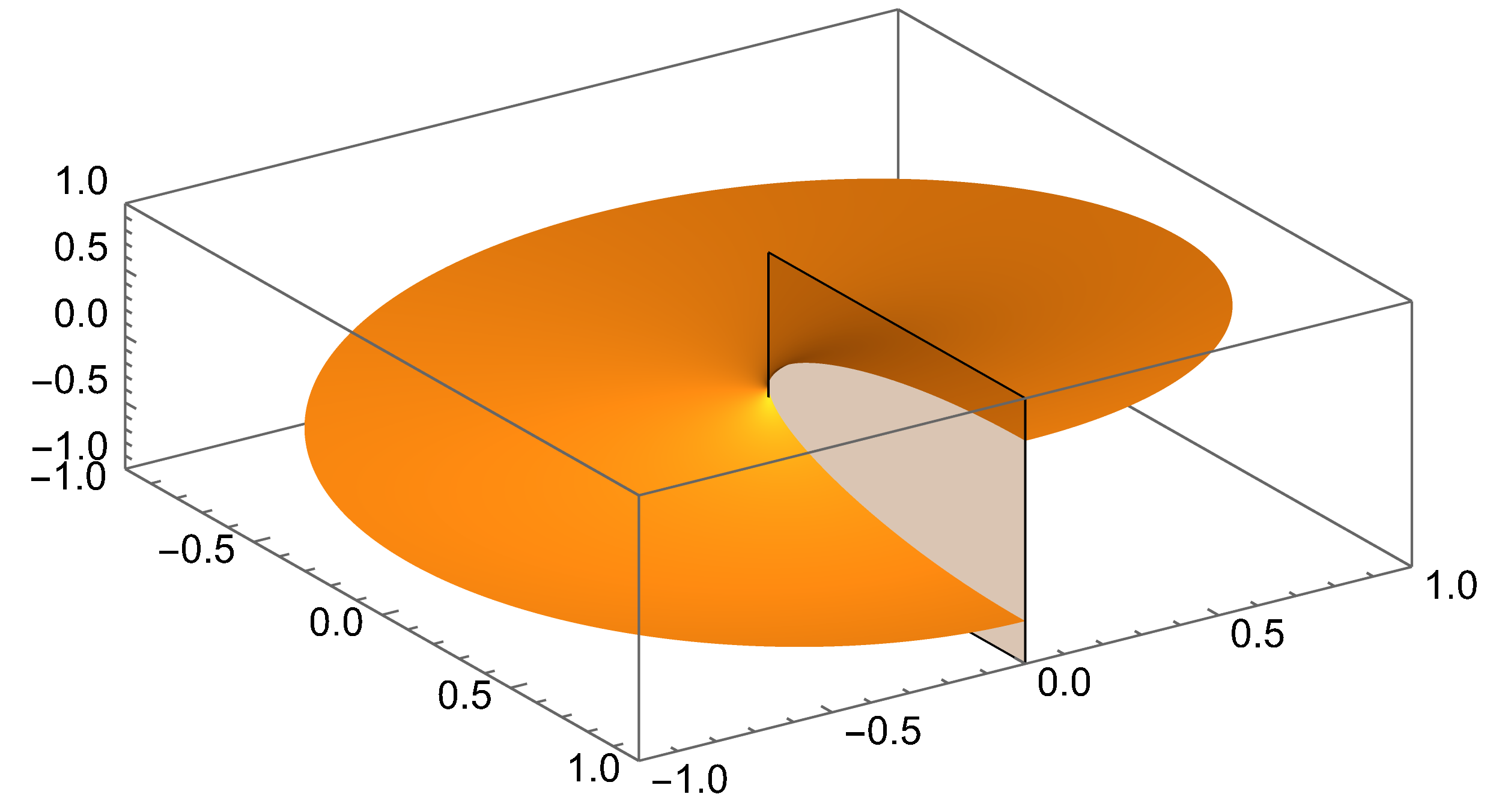}\kern6mm\includegraphics[width=63mm]{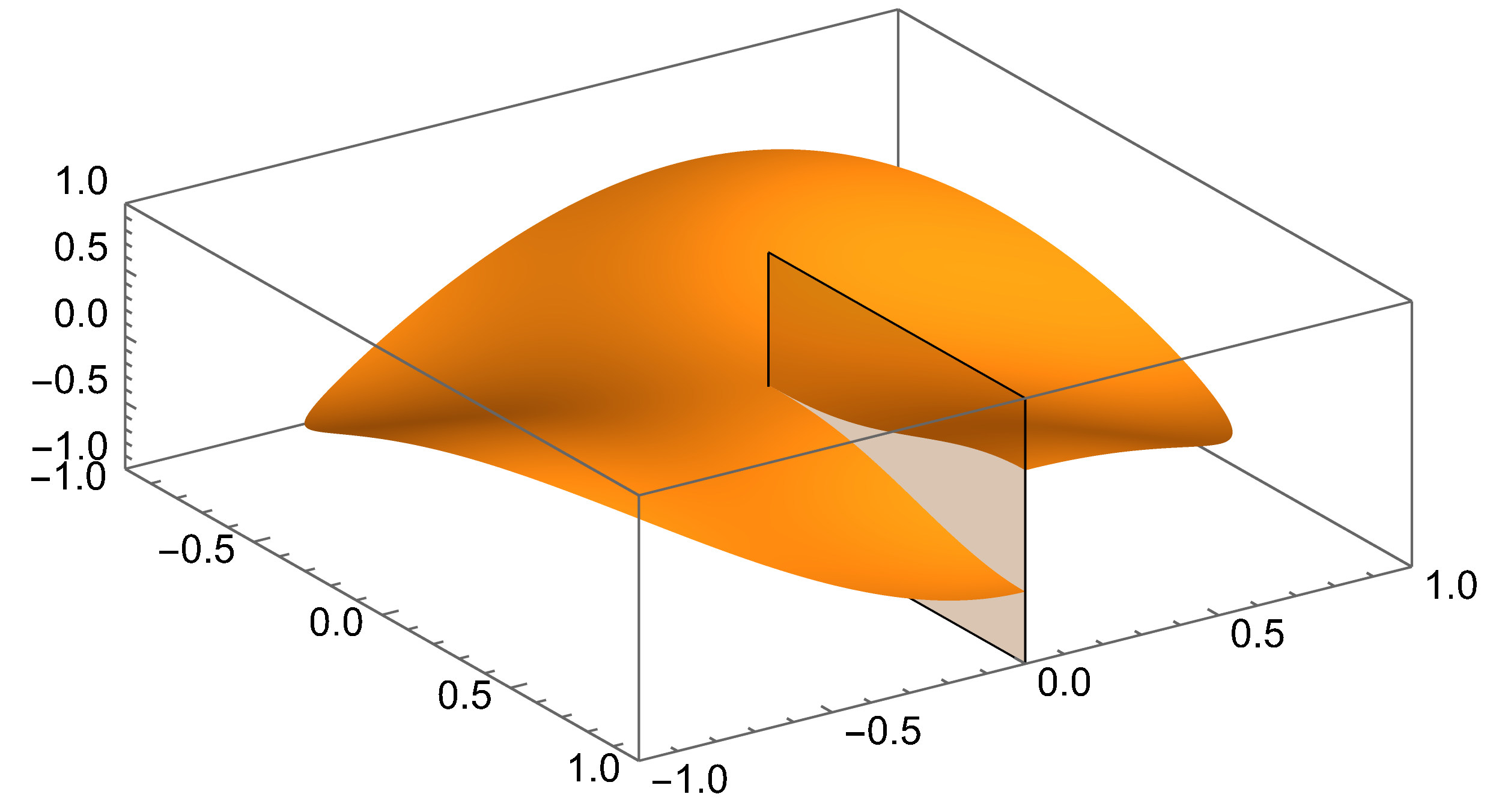}}
 \caption{If $h = 1/10$ in the cylinder with $D = \{ |x| < 1 \} \setminus \{ x_1 \geq 0,
 \ x_2 = 0 \}$ and $d = 1$, then for $\rho \approx 4.651901...$ we have $\nu_1^{(+)}
 (\rho) = \nu_{5}^{(-)} (\rho) \approx 0.923307...$ The graphs of traces on $D$ are
 plotted for the corresponding eigenfunctions \eqref{eigen}: $v_1 (x)$ (left) and $v_{5}
 (x)$ (right).}
  \label{fig:14}
%
\vspace{2.4mm}
%
\centering
 \SetLabels
 \L (0.08*0.14) $x_1$\\
 \L (0.61*0.14) $x_1$\\
 \L (0.345*0.08) $x_2$\\
 \L (0.875*0.08) $x_2$\\
 \L (-0.02*0.58) $v_2$\\
 \L (0.495*0.58) $v_{13}$\\
  \endSetLabels
 \leavevmode\mbox{\kern2.5mm}\AffixLabels{\includegraphics[width=63mm]{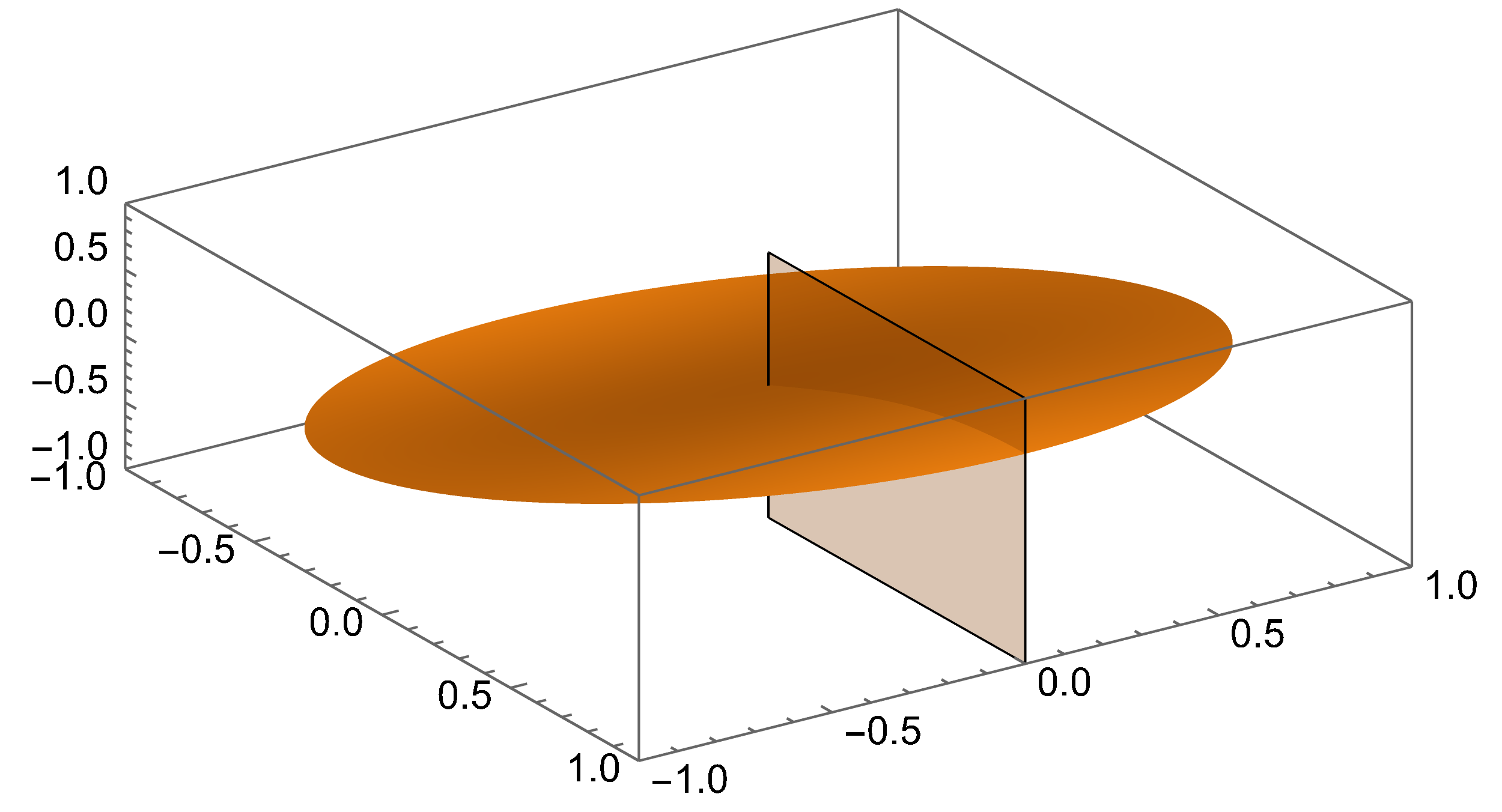}\kern6mm\includegraphics[width=63mm]{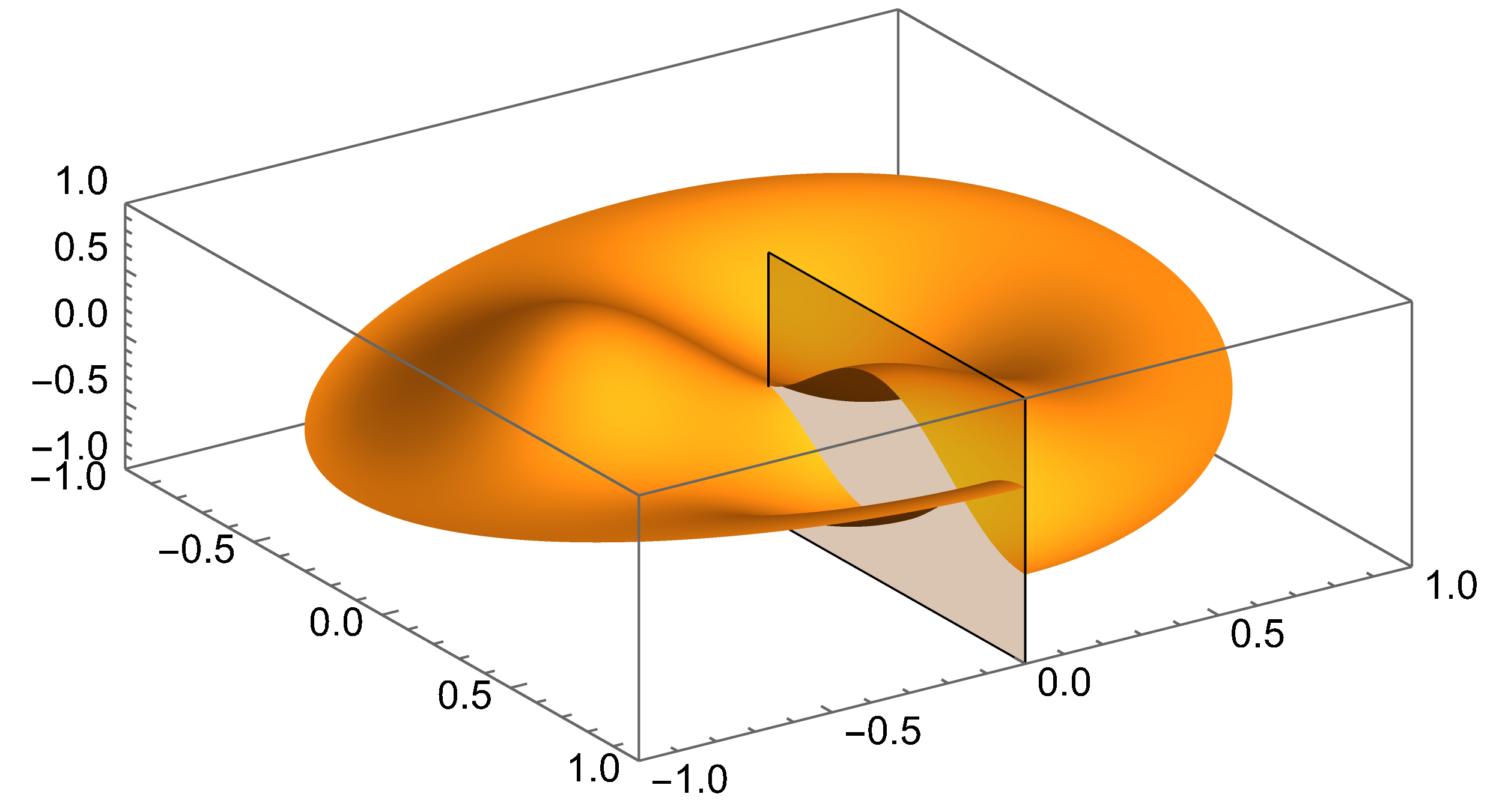}}
 \caption{If $h = 1/10$ in the cylinder with $D = \{ |x| < 1 \} \setminus \{ x_1 \geq 0,
 \ x_2 = 0 \}$ and $d = 1$, then for $\rho \approx 2.750213...$ we have $\nu_2^{(+)}
 (\rho) = \nu_{13}^{(-)} (\rho) \approx 1.729565...$ The graphs of traces on $D$ are
 plotted for the corresponding eigenfunctions \eqref{eigen}: $v_2 (x)$ (left) and
 $v_{13} (x)$ (right).}
  \label{fig:15}
%
\vspace{2.4mm}
%
\centering
 \SetLabels
 \L (0.08*0.14) $x_1$\\
 \L (0.61*0.14) $x_1$\\
 \L (0.345*0.08) $x_2$\\
 \L (0.875*0.08) $x_2$\\
 \L (-0.02*0.58) $v_3$\\
 \L (0.495*0.58) $v_{19}$\\
  \endSetLabels
 \leavevmode\mbox{\kern2.5mm}\AffixLabels{\includegraphics[width=63mm]{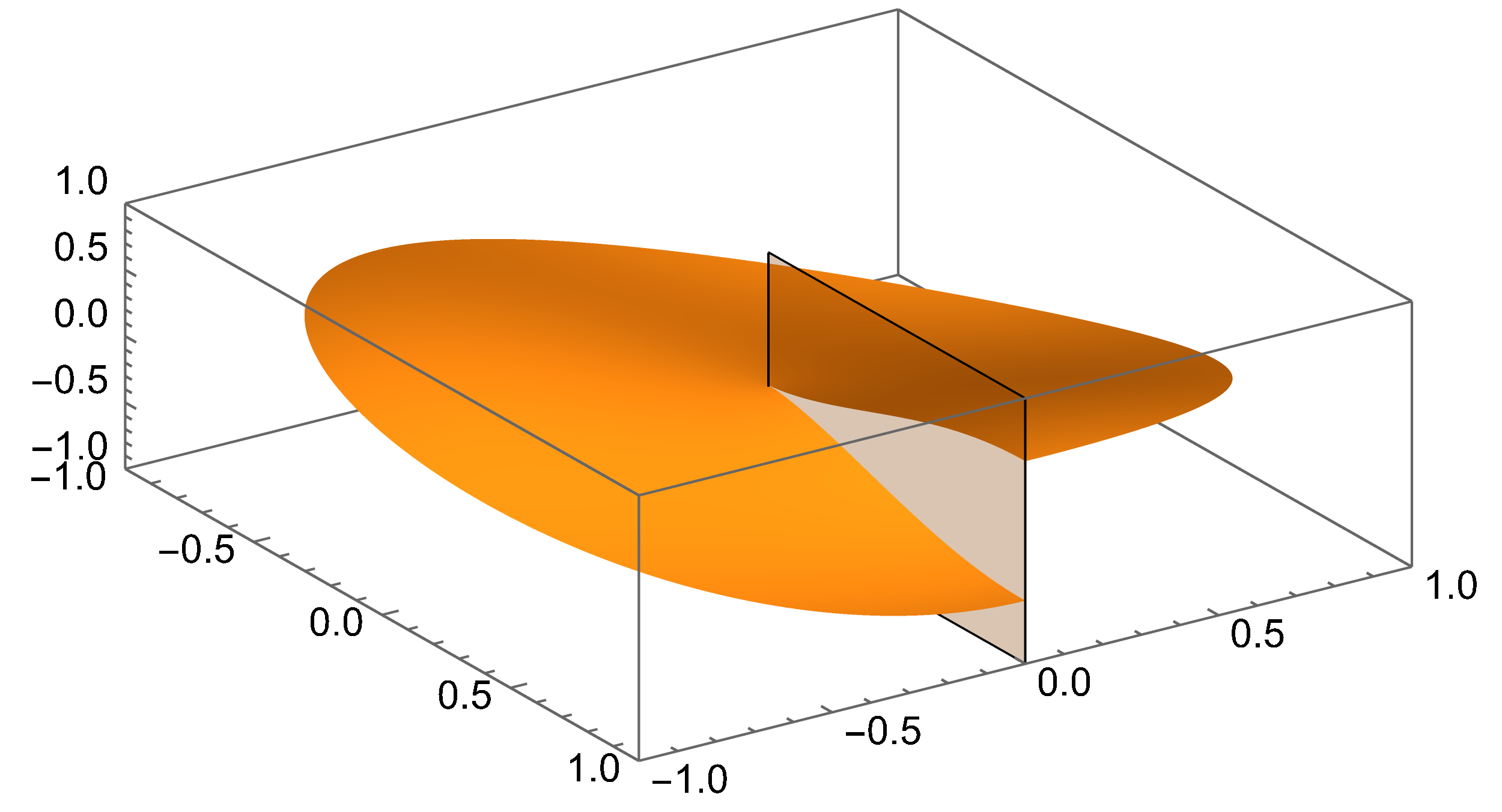}\kern6mm\includegraphics[width=63mm]{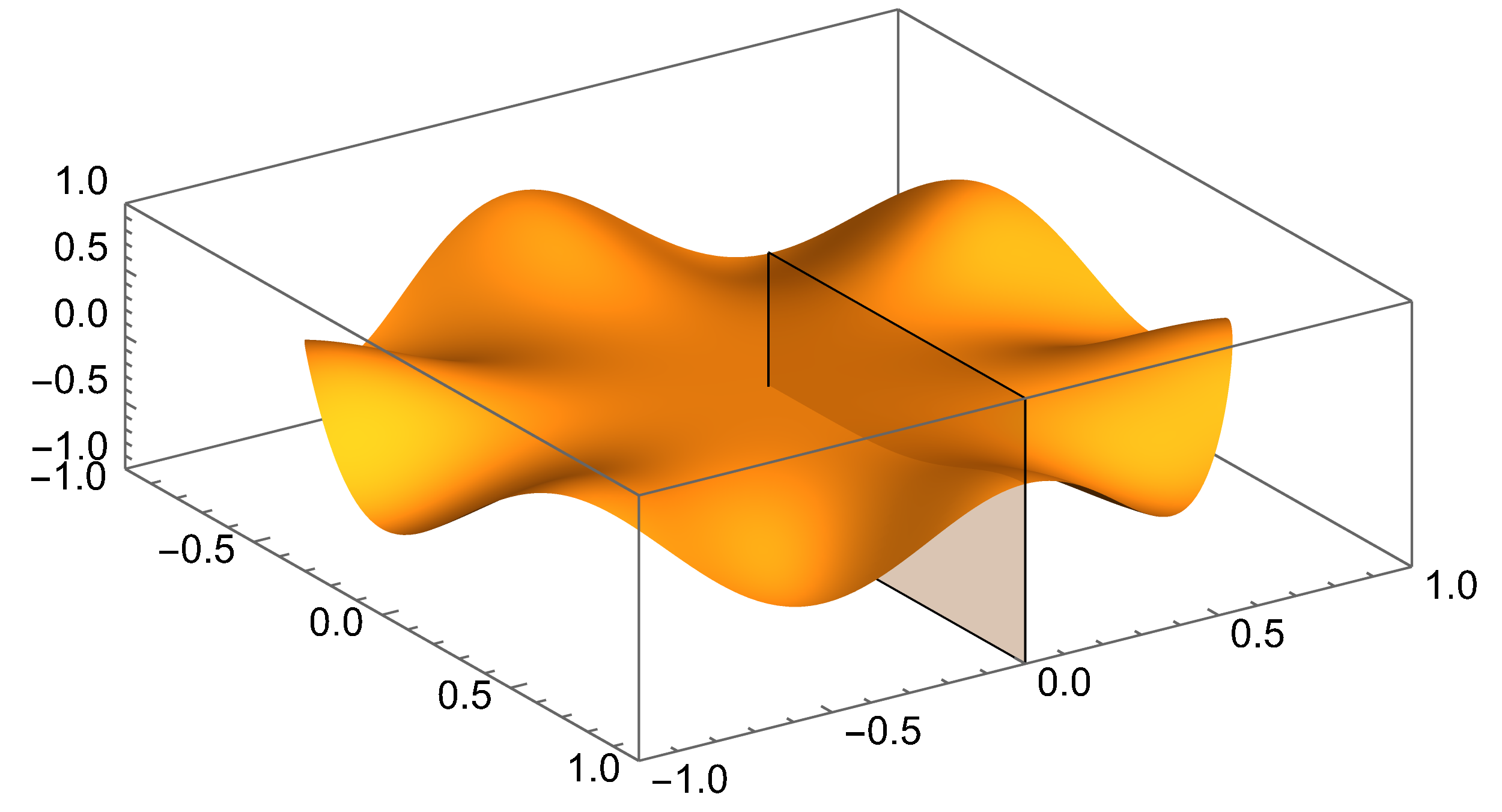}}
 \caption{If $h = 1/10$ in the cylinder with $D = \{ |x| < 1 \} \setminus \{ x_1 \geq 0,
 \ x_2 = 0 \}$ and $d = 1$, then for $\rho \approx  2.678476...$ we have $\nu_3^{(+)}
 (\rho) = \nu_{19}^{(-)} (\rho) \approx 2.413589...$ The graphs of traces on $D$ are
 plotted for the corresponding eigenfunctions \eqref{eigen}: $v_3 (x)$ (left) and $v_{19}
 (x)$ (right).}
  \label{fig:16}
\end{figure}

\subsection{Properties of sloshing eigenfunctions}
\label{sect:3.1}

The same formulae \eqref{eigen}--\eqref{Udef'} (see Subsection~2.3, where the case
of a circular cylinder is considered), but with $\bar{k}_m$ changed to $k_m$,
describe the behaviour of sloshing eigenfunctions in a baffled circular cylinder.
Moreover, the same graphs depict some eigenfunctions which correspond to the same
eigenvalue as in Subsection~2.3, but having different numbers here. For example, the
eigenfunctions
\[ u^{(j)}_{1,+} (x,y) = \bar{k}_1 \, v_1 (x) \, U^{(j)}_{1,+} (y) \ \ \mbox{and} \ \
u^{(j)}_{11,-} (x,y) = \bar{k}_{11} \, v_{11} (x) \, U^{(j)}_{11,-} (y) \, , \ \ j =
1,2 ,
\]
(see Figs.~\ref{fig:4} and \ref{fig:7}, where $U^{(j)}_{1,+}$, $U^{(j)}_{11,-}$ and
$v_1$, $v_{11}$, respectively, are plotted) correspond to $\nu_1^{(+)}
(1.804998...) = \nu_{11}^{(-)} (1.804998...) \approx 1.736930...$ in the case of the
circular cylinder with $D = \{ |x| < 1 \}$, $h = 1/10$ and $d = 1$. However, if $D =
\{ |x| < 1 \} \setminus \{ x_1 \geq 0, \ x_2 = 0 \}$, whereas $h$ and $d$ are the
same, then these eigenfunctions correspond to $\nu_2^{(+)} = \nu_{21}^{(-)} \approx
1.736930...$




Let us consider further examples of the behaviour of the amplitude factor
$U^{(j)}_{m,\pm} (y)$; see the marked intersections of different lines in
Fig.~\ref{fig:10}:
\begin{align*} 
& \mbox{the case} \ \nu_1^{(+)} = \nu_{5}^{(-)} \approx 0.923307... \ \mbox{is
presented in Fig.~\ref{fig:11};} \\ & \mbox{the case} \ \nu_2^{(+)} = \nu_{13}^{(-)}
\approx 1.729565... \ \mbox{is presented in Fig.~\ref{fig:12};} \\ & \mbox{the case}
\ \nu_3^{(+)} = \nu_{19}^{(-)} \approx 2.413589... \ \mbox{is presented in
Fig.~\ref{fig:13}.}
\end{align*}
The most significant distinction of these figures from
Figs.~\ref{fig:4}--\ref{fig:6} is the visible jump in Fig.~\ref{fig:11}:
\[ U^{(1)}_{1,+} (-h) - U^{(2)}_{1,+} (-h) \approx 0.043123... \, ,
\]
whereas the corresponding jumps in Figs.~\ref{fig:4}--\ref{fig:6} are tiny.

On the other hand, there are essential distinctions in the behaviour of
eigenfunctions' traces in horizontal cross-sections of the cylinder with a baffle
and without it; compare Figs.~\ref{fig:14}--\ref{fig:16} and
\ref{fig:7}--\ref{fig:9}, respectively. Namely, even smooth surfaces (they
correspond, for example, to $k_2 = j'_{1,1}$, $k_{19} = j'_{6,1}$ plotted in
Figs.~\ref{fig:15}/\ref{fig:16} (left/right), respectively, and to other values $k_m
= j'_{p,s}$ with integer $p$) are cut along the baffle. If $p$ is half-integer in
$k_m = j'_{p,s}$, then the corresponding $v_m$ has a jump across the baffle; see
Figs.~\ref{fig:14}, \ref{fig:15} (right) and \ref{fig:16} (left).

\section{Conclusions}
\label{sect:4}

We have considered the sloshing problem for a two-layer fluid occupying an open
container of finite depth. The most important results of the present work are
summarized here. First, variational principle is formulated along with its
corollary concerning inequality between the fundamental sloshing eigenvalues for
homogeneous and two-layer fluids occupying the same bounded domain.

The further part of the paper deals with theoretical and numerical analyses of
eigenvalues for containers with vertical walls and horizontal bottoms having
arbitrary cross-sections. The main results for this setting are as follows:

(1) Two families of eigenvalues $\bigl\{ \nu_n^{(+)} \bigr\}_1^\infty$ and $\bigl\{
\nu_n^{(-)} \bigr\}_1^\infty$ are obtained in explicit form in terms of those for
the Neumann Laplacian in a horizontal cross-section of the container. The
eigenvalues $\nu_n^{(+)}$ behave similar to those that describe sloshing in a
homogeneous fluid, whereas the other family $\bigl\{ \nu_n^{(-)} \bigr\}_1^\infty$
includes sufficiently small values provided the ratio of densities is close to
unity.

(2) Various properties of these sequences of eigenvalues are established; in
particular, the dependence on the interface depth $h$ and the ratio of densities
$\rho$ as well as their high frequency asymptotics. Monotonicity of $\nu_n^{(-)}
(\rho)$ and $\nu_n^{(+)} (\rho)$ is proved. Moreover, an interesting property is
observed: if the ratio of densities and the depth $d$ are fixed, then the graph of
$\nu_n^{(\pm)} (h)$ is symmetric about the line $h = d/2$. Finally, the asymptotic
behaviour of the eigenvalue counting function is found.

(3) The behaviour of eigenvalues and the corresponding eigenfunctions is illustrated
numerically for circular cylinders. The graphs of $\nu_m^{(+)} (\rho)$ and
$\nu_m^{(-)} (\rho)$ are presented for initial values of $m$. The multiplicity of
eigenvalues is discussed. Namely, every eigenvalue is either simple or has
multiplicity two in the case of a homogeneous fluid in a circular cylinder. In the
case of a two-layer fluid, the same is true on every curve belonging either to
$\bigl\{\nu^{(+)}_m (\rho) \bigr\}_1^\infty$ or to $\bigl\{ \nu^{(-)}_m (\rho)
\bigr\}_1^\infty$, except for the points, where curves of these two families
intersect each other. Therefore, unlike the case of a homogeneous fluid, there are
eigenvalues of multiplicity three and four. Sloshing eigenfunctions are investigated
for the values of $\rho$, where these multiplicities occur. Another important
property of the two-layer sloshing concerns the lowest eigenvalue, which is much
smaller than that in the case of a homogeneous fluid.

(4) The effect of a radial baffle (it extends throughout the two-layer fluid depth)
on sloshing in a circular container is also demonstrated. The corresponding results
of computations are compared with those when the depth of the interface is the same
in the cylinder without baffle. It is noteworthy that the reduction effect in the
value of the fundamental eigenvalue due to the presence of a baffle manifests itself
more strongly in a two-layer fluid than in a homogeneous one.

\end{document}